\documentclass[preprint,12pt]{elsarticle}

\usepackage{amssymb}
\usepackage{amsmath}
\usepackage{xspace}

\usepackage{graphicx}
\usepackage{caption}
\usepackage{hyperref}
\hypersetup{pdfstartpage=18}  % 指定从第2页开始打开

\usepackage{enumitem}

\usepackage{booktabs}
\usepackage{multirow}
\usepackage{array}

\usepackage{graphicx}
\usepackage{svg}
\usepackage{subfig}  % 用于并排放置图形

\usepackage{amsmath}
\usepackage{amssymb}
\usepackage{amsthm}
\usepackage{bm}
\usepackage{bbm}
\usepackage[linesnumbered,ruled,vlined]{algorithm2e}
\usepackage{newtxmath}
\usepackage{amsfonts} % 提供额外的数学字体

\usepackage{balance}

\usepackage{graphicx} % 用于插入图像
\usepackage{float}    % 用于控制浮动体位置
\usepackage{caption}  % 用于控制标题的格式

\newtheorem{theorem}{Theorem}
\newtheorem{lemma}[theorem]{Lemma}

\newtheorem{definition1}{Definition}

\newcommand{\BioVSS}{\texttt{BioVSS++}\xspace}

\newcommand{\NaiveBioVSS}{\texttt{BioVSS}\xspace}

\newcommand{\BDLCF}{\texttt{BioFilter}\xspace}

\newcommand{\BioHash}{\texttt{BioHash}\xspace}

\newcommand{\LSH}{\texttt{LSH}\xspace}

\newcommand{\DESSERT}{\texttt{DESSERT}\xspace}

\newcommand{\IVFFLAT}{\texttt{IVFFLAT}\xspace}

\newcommand{\HNSW}{\texttt{HNSW}\xspace}

\newcommand{\IndexIVFPQ}{\texttt{IndexIVFPQ}\xspace}

\newcommand{\IVFScalarQuantizer}{\texttt{IVFScalarQuantizer}\xspace}

\newcommand{\WTA}{\texttt{WTA}\xspace}

\newcommand{\FlyHash}{\texttt{FlyHash}\xspace}

\newcommand{\CS}{\textsf{CS}\xspace}

\newcommand{\Medicine}{\textsf{Medicine}\xspace}

\newcommand{\Picture}{\textsf{Picture}\xspace}

\newcommand{\COO}{\texttt{COO}\xspace}
\newcommand{\CSR}{\texttt{CSR}\xspace}

\newcommand{\myparagraph}[1]{\addvspace{0.3\baselineskip}\noindent\textbf{#1.}~}

\journal{x}

\begin{document}

\begin{frontmatter}

%% Title, authors and addresses

%% use the tnoteref command within \title for footnotes;
%% use the tnotetext command for theassociated footnote;
%% use the fnref command within \author or \affiliation for footnotes;
%% use the fntext command for theassociated footnote;
%% use the corref command within \author for corresponding author footnotes;
%% use the cortext command for theassociated footnote;
%% use the ead command for the email address,
%% and the form \ead[url] for the home page:
%% \title{Title\tnoteref{label1}}
%% \tnotetext[label1]{}
%% \author{Name\corref{cor1}\fnref{label2}}
%% \ead{email address}
%% \ead[url]{home page}
%% \fntext[label2]{}
%% \cortext[cor1]{}
%% \affiliation{organization={},
%%             addressline={},
%%             city={},
%%             postcode={},
%%             state={},
%%             country={}}
%% \fntext[label3]{}

\title{Approximate Vector Set Search Inspired by Fly Olfactory Neural System}

%% use optional labels to link authors explicitly to addresses:
%% \author[label1,label2]{}
%% \affiliation[label1]{organization={},
%%             addressline={},
%%             city={},
%%             postcode={},
%%             state={},
%%             country={}}
%%
%% \affiliation[label2]{organization={},
%%             addressline={},
%%             city={},
%%             postcode={},
%%             state={},
%%             country={}}

%% 机构信息定义
\affiliation[whu-cs]{organization={School of Computer Science, Wuhan University},
                    city={Wuhan},
                    postcode={430072},
                    state={Hubei},
                    country={China}}

\affiliation[whu-bdi]{organization={Big Data Institute, Wuhan University},
                     city={Wuhan},
                     postcode={430072},
                     state={Hubei},
                     country={China}}

\affiliation[amazon]{organization={Amazon.com, Inc.},
                    city={Seattle},
                    postcode={98109},
                    state={WA},
                    country={USA}}

%% 作者信息（使用标签链接到对应机构）
\author[whu-cs]{Yiqi Li}

\author[whu-cs]{Sheng Wang} %% 通讯作者

\author[amazon]{Zhiyu Chen}

\author[whu-cs]{Shangfeng Chen}

\author[whu-cs,whu-bdi]{Zhiyong Peng} %% 通讯作者

%% Abstract
\begin{abstract}
Vector set search is an underexplored paradigm of similarity search that focuses on retrieving vector sets similar to a given query set. This paradigm captures the intrinsic structural alignment between sets and real-world entities, enabling finer-grained and more consistent relationship modeling across a wide range of applications. However, it faces significantly greater efficiency challenges than traditional single-vector search, primarily due to the combinatorial explosion of set-to-set comparisons. In this work, we address both the combinatorial complexity inherent in vector set comparisons and the curse of dimensionality carried over from single-vector search. To overcome these obstacles, we propose an efficient algorithm for vector set search, \NaiveBioVSS\footnote{This work is an extended version of our paper titled \textit{Approximate Vector Set Search: A Bio-Inspired Approach for High-Dimensional Spaces}, presented at the 41st IEEE International Conference on Data Engineering (ICDE 2025).} (\textbf{B}io-inspired \textbf{V}ector \textbf{S}et \textbf{S}earch). \NaiveBioVSS draws inspiration from the fly olfactory neural system, which employs sparse coding principles for efficient odor recognition in biological neural networks. Based on this biological mechanism, our approach quantizes vectors into sparse binary codes and employs a Bloom filter-based index leveraging set membership properties. Experimental results show that \NaiveBioVSS achieves over 50× speedup compared to brute-force linear scanning on million-scale datasets, while maintaining a high recall of up to 98.9\%, demonstrating its effectiveness for large-scale vector set search.
\end{abstract}


\end{frontmatter}

%% Add \usepackage{lineno} before \begin{document} and uncomment 
%% following line to enable line numbers
%% \linenumbers

%% main text
%%

\section{Introduction}
Vector search constitutes a core computational task in numerous fields, including information retrieval \cite{huang2013learning}, recommender systems \cite{nigam2019semantic}, and computer vision \cite{johnson2019billion}. Most existing methods focus on single-vector queries and aim to address challenges such as the curse of dimensionality \cite{indyk1998approximate} and the sparsity of high-dimensional vector spaces \cite{zhao2023towards}. In this work, we investigate a novel and challenging problem termed \textit{Vector Set Search}. Given a query vector set $\mathbf{Q}$, the goal is to retrieve the top-$k$ most relevant vector sets from a database $\mathbf{D} = \{\mathbf{V}_1, \mathbf{V}_2, \dots, \mathbf{V}_n\}$. This problem extends beyond conventional single-vector search \cite{ICDE_cluster_ann} by involving intra-set vector matching and aggregation, leading to increased computational complexity. At present, there are no efficient solutions available to handle this complexity effectively.

\noindent \textbf{Applications.} Vector set search supports a wide range of applications due to its inherent structural alignment with real-world entities, including:

\begin{itemize}[noitemsep, leftmargin=*]
\item \textit{Academic Entity Discovery.} Academic entity discovery supports efficient navigation of the rapidly expanding scholarly landscape. By modeling academic entities as vector sets—such as representing a paper’s profile \cite{yin2021mrt} through vectors corresponding to its citations, or characterizing an author’s profile \cite{topic_trajectory} as a vector set derived from their publications—vector set search facilitates the identification of related research and the analysis of academic trends.

\item \textit{Multi-Modal Query Answering.} Multi-modal query answering systems benefit from vector set search to achieve more precise retrieval. As illustrated in \cite{wang2024interactive}, multi-modal inputs (e.g., text, images, and audio) are encoded into vector sets using a unified multi-modal encoder \cite{liang2022mind}. Vector set search then enables cross-modal retrieval by aligning semantic representations across modalities, thereby enhancing the accuracy of the query results.

\item \textit{Multi-Vector Ranking for Recommendation.} In recommendation systems, multi-vector ranking leverages vector set representations to model users. As shown in \cite{Profiles_user}, users are represented as vector sets that encode different facets of user behavior—such as browsing activities, search histories, and click interactions. This vector set-based ranking framework enables flexible matching across diverse user preferences.
\end{itemize}

\begin{figure}[t]
    \centering
    \captionsetup{aboveskip=2pt}
    \includegraphics[width=0.8\textwidth]{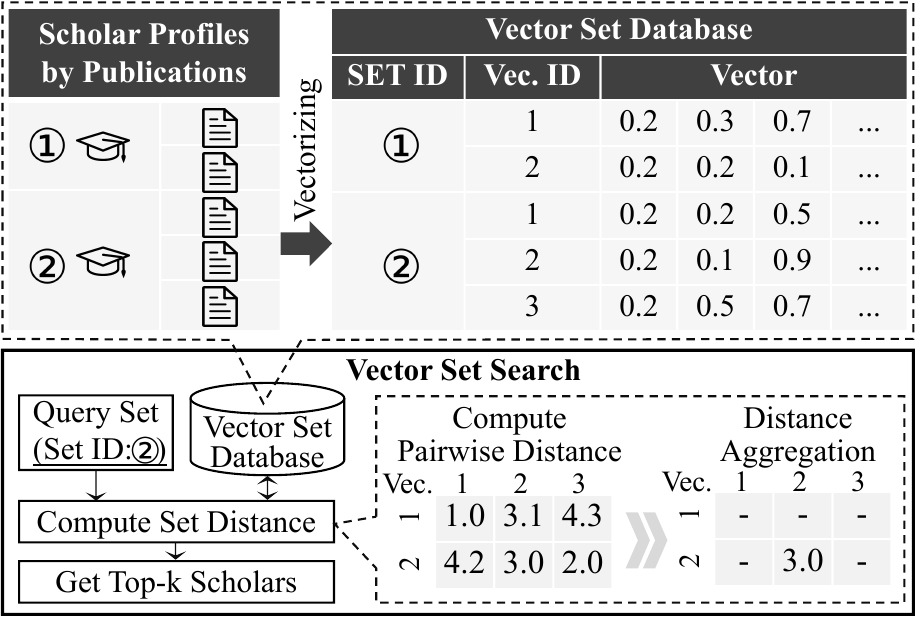}
    \caption{An Exemplar Data Flow of Vector Set Search}
    \label{fig: example}
    % \addvspace{-15.5pt}
\end{figure}

\vskip -0.2em % 通过设置负值来减小空白
Let us delve into scholar search, a characteristic example of vector set search that we examine as our main experimental focus.

% 细化下应用；
% top$-k$需要斜体

%\vskip -0.2em % 通过设置负值来减小空白
\noindent \textit{\underline{Example 1:}} \textit{Figure~\ref{fig: example} presents a representative data flow of vector set search in the context of scholar search \cite{topic_trajectory}, following the construction of a vector set database. The database preparation stage involves encoding scholar profiles as vector sets based on their publication records \cite{bert}, forming a structured database of scholars. The vector set search process consists of two primary phases: distance computation and neighbor retrieval. In the distance computation phase, pairwise distances (e.g., Euclidean distance) between the query vector set and each candidate set in the database are computed and subsequently aggregated (e.g., via $max$ and $min$ operations) into a set-level distance. The neighbor retrieval phase then identifies the top-$k$ most similar scholars according to these aggregated distances. This procedure facilitates the discovery of relevant scholars to support downstream research tasks.}

\noindent \textbf{Motivation.} The preceding example highlights an important observation: the structure of vector sets exhibits strong correspondence with real-world entities, enabling practical and effective applications. With the recent progress in embedding models \cite{bert} capable of generating high-dimensional, semantically rich representations, the range of applicable scenarios has significantly broadened. Nonetheless, the computational complexity arising from high-dimensional vectors and the combinatorial growth of pairwise comparisons \cite{aumuller2023solving} presents notable challenges for building efficient search solutions.

Recent works such as \DESSERT \cite{DESSERT} have adopted hash table-based techniques to accelerate certain types of non-metric distance calculations. However, accurately assessing vector set similarity requires more than speed—it demands reliable discrimination and consistency across comparisons. These fundamental properties hinge on the choice of distance metrics (analyzed in Section~\ref{sec: motivation_hasusdorff}). The Hausdorff distance \cite{hausdorff1914grundzuge, henrikson1999completeness}, by design, satisfies these criteria, offering a mathematically sound means for evaluating set-to-set similarity while preserving essential symmetric properties critical for dependable comparisons.

Driven by these foundational needs in set-level similarity measurement, we emphasize the importance of using metric distances for maintaining evaluation consistency. The Hausdorff distance, rooted in classical set theory \cite{hausdorff1914grundzuge, henrikson1999completeness}, inherently supports set-based distance computation. However, its reliance on exhaustive pairwise distance evaluation renders it computationally expensive, reinforcing the necessity of more efficient approaches for vector set search.

\noindent \textbf{Challenges.} Enabling efficient vector set search in high-dimensional settings under Hausdorff distance introduces several core challenges. The first is the well-known \textit{curse of dimensionality} \cite{indyk1998approximate}. While some recent approaches attempt to mitigate this issue through hash table construction \cite{DESSERT} or by reducing vector sets to single-vector representations \cite{set_to_single}, these techniques are inherently limited to specific non-metric distances, thereby restricting their suitability for Hausdorff-based computations. The second challenge lies in the \textit{aggregation complexity}. Existing algorithms for Hausdorff distance search \cite{wang2021survey, 2021_POINT_SETS, nutanong2011incremental} typically exploit geometric properties to reduce the cost of aggregation. However, such geometric heuristics become ineffective as dimensionality increases, limiting their practical value in high-dimensional scenarios.

\noindent \textbf{Contributions.} To overcome the aforementioned challenges, we propose an approximate top-$k$ vector set search algorithm. Our method is inspired by locality-sensitive hashing techniques \cite{FlyLSH_dasgupta2017neural, FlyLSH_ryali2020bio}, which are in turn motivated by the biological olfactory circuit \cite{luo2010generating}. By emulating the mechanisms of the olfactory circuit, we quantize input vectors into sparse binary codes, thereby mitigating the curse of dimensionality. On top of this, we design an indexing structure based on Bloom filters, leveraging their set membership properties to narrow down the search space and reduce the cost of aggregation. Our main contributions are outlined below:

\begin{itemize}[leftmargin=*, topsep=0pt]
\item We formulate the first approximate vector set search problem in high-dimensional spaces using Hausdorff distance, a natural metric for measuring set similarity (see Section~\ref{sec_problem}).
\item We introduce \NaiveBioVSS, an algorithm that exploits locality-sensitive properties to accelerate vector set search, and provide formal theoretical analysis with proofs to establish its correctness (see Section~\ref{Sec_NaiveBioVSS}).
\item We further develop \BioVSS, an improved version of \NaiveBioVSS that incorporates a dual-layer cascaded filtering mechanism combining inverted index and vector set sketches to minimize unnecessary candidate scans (see Section~\ref{Sec_BioVSS}).
\item We perform comprehensive experiments demonstrating that our method achieves more than 50× speedup over linear scan baselines on million-scale datasets, while maintaining a recall rate of up to 98.9\%, confirming its practical efficiency (see Section~\ref{sec:Experiments}).
\end{itemize}

\section{RELATED WORK}
\label{sec: relatedWork}

\myparagraph{Single-Vector Search}
With the advancement of embedding models, single-vector search has emerged as a widely adopted unified paradigm \cite{zhang2023vbase}. Nonetheless, achieving high efficiency remains a significant bottleneck in real-world applications. To this end, numerous approximate search strategies \cite{ANNsurvey_2020} have been introduced to accelerate single-vector retrieval. These strategies are broadly classified into three categories: locality-sensitive hashing \cite{FlyLSH_ryali2020bio, DevFly, PM-LSH}, graph-based techniques \cite{NavigableSmallWorldGraphs}, and space partitioning schemes \cite{Tree_Indexes, partition_BLISS}. \LSH approaches like \FlyHash \cite{FlyLSH_ryali2020bio} aim to assign similar vectors into identical hash buckets. Graph-based solutions, such as \HNSW \cite{NavigableSmallWorldGraphs}, create navigable proximity graphs to support rapid query traversal. Partitioning-based methods, including \IVFFLAT \cite{zhang2023vbase, douze2024faiss}, \IndexIVFPQ \cite{ge2013optimized, jegou2010product, douze2024faiss}, and \IVFScalarQuantizer \cite{douze2024faiss}, utilize \texttt{k-means} clustering \cite{DBLP_WangSB20} to divide the space and construct inverted indices for faster search. However, these techniques are tailored specifically for single-vector queries, and are not designed to accommodate set-based scenarios.

\myparagraph{Vector Set Search}
Vector set search has garnered relatively limited attention and remains underexplored. Earlier studies primarily focused on low-dimensional vector spaces, approaching vector set search as a set matching problem solvable via algorithms such as the Kuhn-Munkres method \cite{kuhn1955hungarian}. In geospatial research, vector sets representing trajectories were treated as point collections and measured using Hausdorff distance \cite{Polygons_hausdorff, wang2021survey, DBLP_shengBCXLQ18}. Techniques were proposed to improve efficiency, such as integrating road network structures \cite{Pattern_Networks}, applying branch-and-bound algorithms \cite{nutanong2011incremental}, and estimating Hausdorff distances for speedup \cite{2021_POINT_SETS}. Nevertheless, these solutions were designed for low-dimensional contexts and do not translate effectively to high-dimensional domains. The rise of embedding models allows vector sets in high-dimensional spaces to express rich semantic information. Some approaches convert a set of vectors into a single vector \cite{set_to_single}, enabling the application of existing approximate search methods developed for single-vector queries. However, such conversion depends on specific distance functions defined by the corresponding methods. The approach in \cite{DESSERT}, for instance, constructs a hash table to accelerate vector set search, but has limitations when used with certain metrics. Particularly, it lacks theoretical guarantees for Hausdorff distance when the outer aggregation involves the $max$ function. Moreover, current methods provide insufficient support for Hausdorff distance. Our work addresses this gap by focusing on vector set search in high-dimensional space under the Hausdorff distance metric.

\myparagraph{Bio-Inspired Hashing \& Bloom Filter}
Biological olfactory systems across species demonstrate exceptional accuracy in odor recognition \cite{buck1991novel}. For example, mice (rodents) differentiate odors by interpreting neural signals in the olfactory bulb \cite{rhein1980biochemical}, while zebrafish (teleosts) detect amino acids in water through olfactory signaling \cite{pace1985odorant}. Among them, research on flies’ olfactory mechanisms is the most thorough, inspiring the design of sparse locality-sensitive hashing (LSH) functions \cite{FlyLSH_dasgupta2017neural, FlyLSH_ryali2020bio} via computational modeling. Fly-inspired hashing \cite{FlyLSH_dasgupta2017neural}, grounded in the fly olfactory system, offers an LSH technique for single-vector search. In this model, specific odor molecules activate only a small subset of neurons, resulting in sparse activations that facilitate robust odor identification. \BioHash \cite{FlyLSH_ryali2020bio} enhances performance by learning internal data structures to assign connection weights among projection neurons, emulating the biological connectivity of the fly’s olfactory circuit \cite{vosshall2000olfactory}. On a similar note, Bloom filters support efficient set representation by mapping elements to sparse arrays via hash functions. The binary Bloom filter \cite{vosshall2000olfactory} employs a bit array with binary states, while the counting Bloom filter \cite{bonomi2006improved} refines this with count-based granularity. This work integrates the fly olfactory circuit’s locality-sensitive behavior with the structural advantages of Bloom filters. We utilize these parallels to build an efficient indexing structure.

% \begin{table}[!tbp]
%     \centering
%     \caption{The Method of Similarity Search with Vector}
%     \vspace{-1em} 
%     \resizebox{\columnwidth}{!}{%
%     \begin{tabular}{l>{\centering\arraybackslash}p{2cm}>{\centering\arraybackslash}p{2cm}>{\centering\arraybackslash}p{2cm}}
%         \toprule
%         \textbf{Method} & \textbf{Native Search Type} & \textbf{Theory Maturity} & \textbf{Accel. Techniques} \\
%         \midrule
%         NSW \cite{NavigableSmallWorldGraphs} & Single-vector & Developing & Index \\
%         LSH \cite{indyk1998approximate} & Single-vector & Well-developed & Index \\
%         FlyHash \cite{FlyLSH_dasgupta2017neural, FlyLSH_ryali2020bio} & Single-vector & Well-developed & Quant \\
%         IVF \cite{ge2013optimized, douze2024faiss} & Single-vector & Developing  & Index \\
%         PQ \cite{jegou2010product} & Single-vector & Developing & Quant \\
%         DESSERT \cite{DESSERT} & Vector set & Well-developed & Quant \& Index \\
%         \BioVSS (Proposed) & Vector set & Well-developed & Quant \& Index \\
%         \bottomrule
%     \end{tabular}
%     }
%     \label{tab: literature}
% \end{table}
% \setlength{\textfloatsep}{0cm}

\section{{DEFINITIONS \& PRELIMINARIES}}
This section presents the problem formulation and analyzes why the Hausdorff distance is suitable for vector set comparison.

\label{sec_problem}
\subsection{Problem Definition}
\begin{definition1}[\textbf{Vector}]
A \textbf{vector} $\mathbf{v} = (v_1, v_2, ..., v_d)$ is a tuple of real numbers,  in which $v_i \in \mathbb{R}$ and $d \in \mathbb{N}^+$ represents the dimensionality of the vector.
\end{definition1}

\begin{definition1}[\textbf{Vector Set}]
A \textbf{vector set} $\mathbf{V} = \{\mathbf{v}_1, \mathbf{v}_2, ..., \mathbf{v}_m\}$ contains a set of vectors, where $m \in \mathbb{N}^+$ is number of vectors.
\end{definition1}

\begin{definition1}[\textbf{Vector Set Database}]
A \textbf{vector set database} $\mathbf{D} = \{\mathbf{V}_1, \mathbf{V}_2, ..., \mathbf{V}_n\}$ is a collection of vector sets, where each $\mathbf{V}_i$ is a vector set and $n \in \mathbb{N}^+$ is the number of vector sets.
\end{definition1}

\begin{definition1}[\textbf{Hausdorff Distance}] Given two vector sets $\mathbf{Q}$ and $\mathbf{V}$, the \textbf{Hausdorff distance} from $\mathbf{Q}$ to $\mathbf{V}$ is defined as:
\begin{equation*}
Haus(\mathbf{Q}, \mathbf{V})=\max \left(\max _{\mathbf{q} \in \mathbf{Q}} \min _{\mathbf{v} \in \mathbf{V}}dist(\mathbf{q}, \mathbf{v}), \max _{\mathbf{v} \in \mathbf{V}} \min _{\mathbf{q} \in \mathbf{Q}}dist(\mathbf{v}, \mathbf{q})\right),
\end{equation*}
where $dist(\mathbf{q}, \mathbf{v})=  \|\mathbf{q} - \mathbf{v}\|_2 $ is the Euclidean distance between vectors $\mathbf{q} \in \mathbb{R}^d$ and $\mathbf{v} \in \mathbb{R}^d$.
\label{def:Hausdorff}
\end{definition1}

To clarify the computational steps involved in calculating the Hausdorff distance, we present a concrete example below.

\noindent \textit{\underline{Example 2:}} \textit{Figure \ref{fig: hausdorff} demonstrates the procedure for computing the Hausdorff distance. Consider two finite vector sets $\mathbf{Q}=\{\mathbf{q}_1,\mathbf{q}_2,\mathbf{q}_3\}$ and $\mathbf{V}=\{\mathbf{v}_1,\mathbf{v}_2\}$ situated in a high-dimensional space. The Hausdorff distance $Haus(\mathbf{Q}, \mathbf{V})$ is obtained via the following operations: (1) compute the maximum over the minimal distances from each vector in $\mathbf{V}$ to all vectors in $\mathbf{Q}$, which yields $\max_{\mathbf{v}\in\mathbf{V}} \min_{\mathbf{q}\in\mathbf{Q}} d(\mathbf{v},\mathbf{q})=2$; (2) compute the maximum over the minimal distances from each vector in $\mathbf{Q}$ to all vectors in $\mathbf{V}$: $\max_{\mathbf{q}\in\mathbf{Q}} \min_{\mathbf{v}\in\mathbf{V}} d(\mathbf{q},\mathbf{v})=3$; (3) aggregate the two values using the $max$ operator to obtain the final Hausdorff distance: $Haus(\mathbf{Q}, \mathbf{V}) = 3$.}

The computation of the Hausdorff distance entails evaluating all pairwise distances between vectors in the two sets, followed by aggregation as defined in Definition \ref{def:Hausdorff}. This leads to a time complexity of $O(m^2 \cdot d)$, where $m$ denotes the number of vectors in each set and $d$ indicates the dimensionality. The quadratic growth with respect to $m$ renders exact computation impractical for large datasets, particularly when both the set size and dimensionality $d$ are high.

\begin{figure}[t]
    \centering
    \captionsetup{aboveskip=2pt}
    \includegraphics[width=0.8\textwidth]{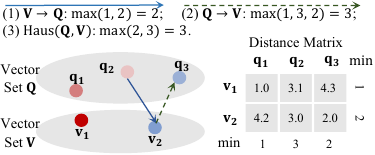}
    \caption{The Calculation Process of Hausdorff Distance}
    \label{fig: hausdorff}
    % \addvspace{-15.5pt}
\end{figure}

To mitigate this computational burden, prior work in single-vector search has demonstrated the effectiveness of approximation strategies, including quantization \cite{FlyLSH_dasgupta2017neural, FlyLSH_ryali2020bio, DevFly} and index pruning \cite{Tree_Indexes, partition_BLISS, NavigableSmallWorldGraphs, ge2013optimized, jegou2010product, douze2024faiss}. Motivated by these insights, we formulate the approximate top-$k$ vector set search problem.

\begin{definition1}[\textbf{Approximate Top-$k$ Vector Set Search}]
\label{ApproximateSearch}
Given a vector set database $\mathbf{D} = \{\mathbf{V}_1, \mathbf{V}_2, ..., \mathbf{V}_n\}$, a query vector set $\mathbf{Q}$ and a Hausdorff distance function $Haus(\mathbf{Q}, \mathbf{V})$, the \textbf{approximate top-k vector set search} is the task of returning $\mathbf{R}$ with probability at least $1-\delta$:
$$
\mathbf{R} = \{ \mathbf{V}^{\star}_1, \mathbf{V}^{\star}_2, ..., \mathbf{V}^{\star}_k \} = \operatorname{argmin}^k_{\mathbf{V}^{\star}_i \in \mathbf{D}} Haus(\mathbf{Q}, \mathbf{V}^{\star}_i),
$$
where $\operatorname{argmin}^k_{\mathbf{V}_i \in \mathbf{D}}$ selects $k$ sets $\mathbf{V}^{\star}_i$ minimizing $Haus(\mathbf{Q}, \mathbf{V}^{\star}_i)$ and $\mathbf{V}^{\star}_k$ has the $k$-th smallest distance from $\mathbf{Q}$. The failure probability is denoted by $\delta \in [0, 1]$.
\end{definition1}

Approximate top-$k$ vector set search achieves higher efficiency by allowing a minor loss in accuracy, effectively balancing computational speed and result precision.

\subsection{Suitability Analysis of Hausdorff Distance}
\label{sec: motivation_hasusdorff}
\begin{figure}[h]
    \centering
    \captionsetup{aboveskip=2pt}
    \includegraphics[width=0.85\textwidth]{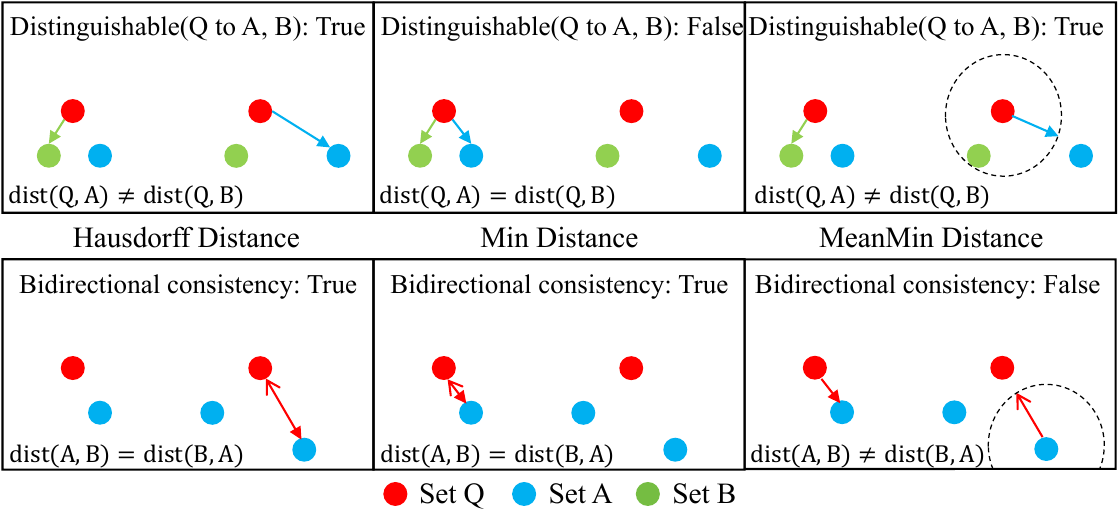}
    \caption{Comparative Analysis of Different Distance Measures}
    \label{fig: meticMotivation}
\end{figure}
To demonstrate the advantages of the Hausdorff distance in vector set comparison, a comparative analysis of three distance measures is conducted: minimum distance, mean minimum distance, and Hausdorff distance. The focus is on illustrating the superior precision and symmetry of the Hausdorff distance.

Three vector sets $\mathbf{Q}$, $\mathbf{A}$, and $\mathbf{B}$ are used for the precision analysis, each containing two vectors. For the symmetry analysis, two vector sets $\mathbf{Q}$ and $\mathbf{A}$ are employed, where $\mathbf{Q}$ contains two vectors and $\mathbf{A}$ contains three vectors. For visualization purposes, all vectors are represented as points in a two-dimensional space. This representation allows for a clear illustration of the distance measures' properties.

The distance measures are defined as follows:
\begin{enumerate}
    \item Minimum (Min) distance: 
          $$d_{min}(\mathbf{A},\mathbf{B}) = \min_{\mathbf{a}\in \mathbf{A}, \mathbf{b}\in \mathbf{B}} d(\mathbf{a},\mathbf{b}),$$
    
    \item Mean Minimum (MeanMin) distance: 
          $$d_{mean-min}(\mathbf{A},\mathbf{B}) = \frac{1}{|\mathbf{A}|}\sum_{\mathbf{a}\in \mathbf{A}} \min_{\mathbf{b}\in \mathbf{B}} d(\mathbf{a},\mathbf{b}),$$
    
    \item Hausdorff distance: 
          As defined in Definition \ref{def:Hausdorff}.
\end{enumerate}

Here, $d(\cdot, \cdot)$ denotes the Euclidean distance, and $|\mathbf{A}|$ denotes the cardinality of vector set $\mathbf{A}$.

To evaluate precision, the distances between \( \mathbf{Q} \) and vector sets \( \mathbf{A} \) and \( \mathbf{B} \) are analyzed, as shown in the first row of Figure \ref{fig: meticMotivation}. The specific distance matrices are presented below:

\begin{equation*}
\begin{array}{c|cc}
\mathbf{Q} \text{ to } \mathbf{A} & \mathbf{Q}_1 & \mathbf{Q}_2 \\
\hline
\mathbf{A}_1 & 1 & 5 \\
\mathbf{A}_2 & 6 & 3
\end{array}
\qquad
\begin{array}{c|cc}
\mathbf{Q} \text{ to } \mathbf{B} & \mathbf{Q}_1 & \mathbf{Q}_2 \\
\hline
\mathbf{B}_1 & 1 & 5 \\
\mathbf{B}_2 & 4 & 1
\end{array}
\end{equation*}

The analysis yields the following results:
\begin{enumerate}
    \item $d_{min}(\mathbf{Q},\mathbf{A}) = d_{min}(\mathbf{Q},\mathbf{B}) = 1$
    \item $d_{mean-min}(\mathbf{Q},\mathbf{A}) = 2$, $d_{mean-min}(\mathbf{Q},\mathbf{B}) = 1$
    \item $d_H(\mathbf{Q},\mathbf{A}) = 3$, $d_H(\mathbf{Q},\mathbf{B}) = 2$
\end{enumerate}

These results clearly demonstrate the superior precision of the Hausdorff distance. The minimum distance fails to differentiate between $\mathbf{Q}$'s relationships with $\mathbf{A}$ and $\mathbf{B}$. The mean minimum distance and Hausdorff distance clearly distinguish both the similarities and differences between the vector sets. While both $d_{mean-min}$ and $d_H$ show discriminative power, we further examine their symmetry properties.

To examine symmetry, a case with two vector sets, $\mathbf{Q}$ and $\mathbf{A}$, where $\mathbf{Q}$ contains two vectors and $\mathbf{A}$ contains three vectors, is analyzed. The distance matrix is:

\begin{equation*}
\begin{array}{c|cc}
\mathbf{Q} \text{ to } \mathbf{A} & \mathbf{Q}_1 & \mathbf{Q}_2 \\
\hline
\mathbf{A}_1 & 1 & 4 \\
\mathbf{A}_2 & 4 & 1 \\
\mathbf{A}_3 & 7 & 3
\end{array}
\end{equation*}

The analysis yields the following results:
\begin{enumerate}
    \item $d_{min}(\mathbf{Q},\mathbf{A}) = d_{min}(\mathbf{A},\mathbf{Q}) = 1$
    \item $d_{mean-min}(\mathbf{Q},\mathbf{A}) = 1$, $d_{mean-min}(\mathbf{A},\mathbf{Q}) = 1.67$
    \item $d_H(\mathbf{Q},\mathbf{A}) = d_H(\mathbf{A},\mathbf{Q}) = 3$
\end{enumerate}

These results highlight the perfect symmetry of the Hausdorff distance. The minimum distance also exhibits symmetry. However, the mean minimum distance produces different results depending on the direction of comparison. The Hausdorff distance maintains consistency regardless of the order of vector set comparison, ensuring reliable and consistent similarity assessments.

The examples demonstrate the superiority of the Hausdorff distance in vector set comparisons. Its advantages are twofold. First, it provides enhanced precision in set differentiation. Second, it maintains consistent symmetry regardless of comparison direction. These properties are not present in other common measures. The Hausdorff distance effectively captures fine-grained differences while ensuring mathematical consistency. This balance makes it particularly suitable for complex vector set comparisons. As a result, it serves as a versatile and reliable metric across various applications.

\section{Proposed BioVSS}
\label{Sec_NaiveBioVSS}
We propose \NaiveBioVSS, a search algorithm optimized for the approximate top-$k$ vector set search task. Inspired by the fly's olfactory circuitry, \NaiveBioVSS exploits its inherent locality-sensitive characteristics to improve search performance. The theoretical basis for the locality-sensitive Hausdorff distance is laid out in Section \ref{sec: theory}.

\subsection{Overview Algorithm}
\label{sec: naiveBioVSS}
\NaiveBioVSS consists of two main components: 1) the binary hash encoding for vector sets, and 2) the search execution in \NaiveBioVSS.

\begin{table}[t]
\centering
\footnotesize 
\caption{Summary of Major Notations}
% \vspace{-1em} % 调整表格与文字之间的距离
\begin{tabular}{ll}
\hline
\textbf{Notation} & \textbf{Description} \\
\hline
$\mathbf{D}$ & Vector Set Database \\
$\mathbf{D}^\mathbf{H}$ & Sparse Binary Codes of $\mathbf{D}$ \\
$\mathbf{T}, \mathbf{Q}$ & Vector Set of Target and Query\\
$n, m , m_q $ & Cardinality of $\mathbf{D}$, $\mathbf{T}$,  and $\mathbf{Q}$\\
$L_{wta}=L$ & Number of Hash Functions / Winner-Takes-All \\
$\mathcal{H}$ & Hash Function (see Definition \ref{def: FlyHash})\\
$k$ &  The Number of Results Returned \\
$\sigma(\cdot)$ & Min-Max Similarity Function (see Lemma \ref{lemma:sigma_bounds}) \\
$S_{ij}^{\alpha}, S_{ij}^{\beta}, \mathbf{S}$ & (i,j)-th Entry of the Similarity Matrix $\mathbf{S}$ \\
$s_{\max}, s_{\min}$ & Max. and Min. Real Similarities btw. Vector Sets \\
$\hat{\mathbf{S}}, \hat{s}_{\max}, \hat{s}_{\min}$ & Estimated Values of $\mathbf{S}, s_{\max}$, and $s_{\min}$ \\
$\delta$ & Vector Set Search Failure Probability \\
$B_{\alpha}^{\star}, B_{\beta}, B^{\prime}, B^{\star}$ & Similarity Bounds of Vector Sets (see Theorem \ref{theorem: ultimate}) \\
\hline
\end{tabular}
\label{tab:notations}
\end{table}

\subsubsection{Binary Hash Encoding for Vector Set}
The rationale for employing binary hash encoding lies in harnessing the similarity-preserving nature of hashing to reduce computational overhead and improve search efficiency.

We begin by formalizing the unified framework of \LSH, which underpins a range of \LSH methods \cite{FlyLSH_dasgupta2017neural, FlyLSH_ryali2020bio, DevFly, PM-LSH}, originally introduced in \cite{indyk1998approximate}.

\begin{definition1}[\textbf{Locality-Sensitive Hashing Function} \cite{indyk1998approximate}]
An \LSH hash function $h: \mathbb{R}^d \rightarrow \mathbb{R}$ is called a similarity-preserving hash for vectors $\mathbf{a}, \mathbf{b} \in \mathbb{R}^d$ if
\[
\mathbb{P}(h(\mathbf{a})=h(\mathbf{b})) = \operatorname{sim}(\mathbf{a}, \mathbf{b}), \forall \mathbf{a}, \mathbf{b} \in \mathbb{R}^d,
\]
where $\operatorname{sim}(\mathbf{a}, \mathbf{b}) \in [0,1]$ is the similarity between vectors $\mathbf{q} \in \mathbb{R}^d$ and $\mathbf{v} \in \mathbb{R}^d$.
\label{def_LSH}
\end{definition1}

\begin{figure}[h]
    \addvspace{-11pt}
    \centering
    \captionsetup{aboveskip=2pt}
    \includegraphics[width=0.85\textwidth]{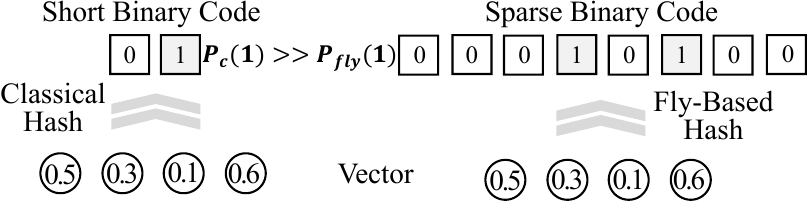}
    \caption{Classical Hashing vs. Fly-Based Hashing}
    \label{fig: tradition}
    % \addvspace{-11pt}
\end{figure}

Building upon this unified paradigm, our framework incorporates a fly-inspired hashing mechanism. As depicted in Figure \ref{fig: tradition}, this method fundamentally contrasts with conventional hashing schemes in its approach to dimensional transformation. While classical hashing methods commonly apply dimensionality reduction to produce dense hash codes, the fly-based strategy draws from the neural projection mechanisms observed in the Drosophila olfactory system \cite{FlyLSH_dasgupta2017neural}, performing expansion via projections and yielding sparse binary representations.

The adoption of fly-based hashing is driven by two main considerations. First, its sparse binary encoding naturally aligns with the set-oriented structure of Bloom filters, which is critical for the functionality of our framework (refer to Section \ref{Sec_BioVSS}). This structural alignment renders it more appropriate than traditional distance-preserving hashing functions. Second, within fly-inspired methods, we employ \BioHash \cite{FlyLSH_ryali2020bio} in preference to the original \FlyHash \cite{FlyLSH_dasgupta2017neural}, as empirical evaluations indicate that \BioHash offers roughly double the similarity preservation performance compared to \FlyHash \cite{FlyLSH_ryali2020bio}.

\begin{definition1}[\textbf{FlyHash} \cite{FlyLSH_dasgupta2017neural} \& \textbf{BioHash} \cite{FlyLSH_ryali2020bio}]
A \texttt{FlyHash} or \texttt{BioHash} is a function $\mathcal{H}: \mathbb{R}^d \rightarrow \{0, 1\}^b$ that maps a vector $\mathbf{v} \in \mathbb{R}^d$ to a sparse binary vector $\mathbf{h} = \mathcal{H}(\mathbf{v}) \in \{0, 1\}^b$. The hash code is generated as $\mathbf{h} = WTA(W\mathbf{v})$, where $W \in \mathbb{R}^{b \times d}$ is a random projection matrix, and $WTA$ is the Winner-Take-All (\WTA) operation that sets the $L_{wta}$ largest elements to 1 and the rest to 0. The resulting binary vector $\mathbf{h}$ has $L_{wta}$ non-zero elements ($\|\mathbf{h}\|_1 = L_{wta} \ll b$), where $L_{wta}$ also represents the number of hash functions composing $\mathcal{H}$ defined in Lemma \ref{def_LSH}.
\label{def: FlyHash}
\end{definition1}

\setlength{\textfloatsep}{0cm}
\begin{algorithm}[!hb]
\small
\caption{\texttt{Gen\_Binary\_Codes}($\mathbf{D}, L_{wta}$)}
\label{alg:hashing}
\KwIn{$\mathbf{D} = \{\mathbf{V}_i\}_{i=1}^n$: vector set database, $L_{wta}$: \# of \WTA.}
\KwOut{$\mathbf{D}^\mathbf{H}$: sparse binary codes.}

\textbf{Initialization:} $\mathbf{D}^\mathbf{H} = \emptyset$

\For{$j = 1$ \KwTo $n$}{
    $\mathbf{H}_j = \emptyset$\;
    \For{each $\mathbf{v} \in \mathbf{V}_j$}{
        $\mathbf{h}_p = W\mathbf{v}$\tcp*{Matrix projection}
        $\mathbf{h} = WTA(\mathbf{h}_p,L_{wta})$\tcp*{\hspace{-2.55mm} Winner-takes-all}
        $\mathbf{H}_j = \mathbf{H}_j \cup \{\mathbf{h}\}$\;
    }
    $\mathbf{D}^\mathbf{H} = \mathbf{D}^\mathbf{H} \cup \{\mathbf{H}_j\}$\;
}

\Return $\mathbf{D}^\mathbf{H}$\;
\end{algorithm}

As shown in Figure \ref{fig: hashing}(a), \FlyHash \cite{FlyLSH_dasgupta2017neural} constructs \LSH codes by simulating the fly’s olfactory neural circuit. It projects input vectors into a higher-dimensional space using a projection matrix that models the behavior of projection neurons, thereby increasing representational richness. Winner-take-all mechanism is then employed to produce sparse binary codes by retaining only the most active neurons. This approach has been proven to satisfy the locality-sensitive property \cite{FlyLSH_dasgupta2017neural} (see Definition \ref{def_LSH}).

\begin{figure}[t]
    \addvspace{-11pt}
    \centering
    \captionsetup{aboveskip=2pt}
    \includegraphics[width=0.8\textwidth]{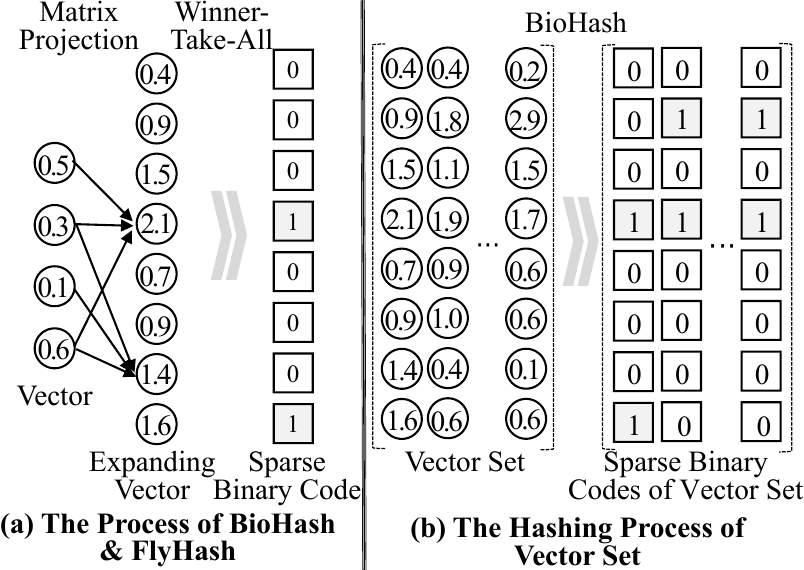}
    \caption{The Hashing Process of \NaiveBioVSS}
    \label{fig: hashing}
    % \addvspace{-11pt}
\end{figure}

Theoretical justifications for this method are provided in Section \ref{sec: theory}, where we demonstrate that the generated hash codes are well-suited to support approximate top-$k$ vector set search. As depicted in Figure \ref{fig: hashing}(b), the vector sets are first quantized into binary hash codes. We now proceed to describe the hashing procedure in detail.

Algorithm \ref{alg:hashing} presents the procedure for computing sparse binary representations for a database of vector sets. The input includes a database $\mathbf{D}$ containing $n$ vector sets $\{\mathbf{V}_i\}_{i=1}^n$ and a parameter $L_{wta}$, which governs the WTA operation. The output $\mathbf{D}^\mathbf{H}$ consists of the corresponding sparse binary codes. The algorithm processes each vector set $\mathbf{V}_j$ in sequence (line 2), applying the \BioHash transformation (refer to Definition \ref{def: FlyHash}) to every vector $\mathbf{v}$ it contains (lines 5–6). This transformation involves computing $W\mathbf{v}$ through a linear projection, followed \WTA operation that retains the top $L_{wta}$ entries as 1 while setting the rest to 0, yielding a binary vector $\mathbf{h}$ with exactly $L_{wta}$ active bits. These codes are aggregated on a per-vector-set basis (line 7) and stored in $\mathbf{D}^\mathbf{H}$ (line 8).

\begin{algorithm}[t]
\small
\caption{\texttt{BioVSS\_Topk\_Search}($\mathbf{Q}, k, \mathbf{D}, \mathbf{D}^\mathbf{H}, L_{wta}$,$c$)}
\label{alg:query}
\KwIn{$\mathbf{Q}$: query vector set, $k$: \# of top sets, $\mathbf{D} = \{\mathbf{V}_i\}_{i=1}^n$: vector set database, $\mathbf{D}^\mathbf{H}$: sparse binary codes, $L_{wta}$: \# of \WTA, $c$: the size of candidate set.} 
\KwOut{$\mathbf{R}$: top-$k$ vector sets.}

\textbf{Initialization:} $\mathbf{Q}^\mathbf{H} = \emptyset$, $\mathcal{C} = \emptyset$\;
\For{each $\mathbf{q} \in \mathbf{Q}$}{
    $\mathbf{h}_p = W\mathbf{q}$\tcp*{Matrix projection}
    $\mathbf{h}_q = WTA(\mathbf{h}_p, L_{wta})$\tcp*{Winner-takes-all}
    $\mathbf{Q}^\mathbf{H} = \mathbf{Q}^\mathbf{H} \cup \{\mathbf{h}_q\}$\;
}
\tcp{Select candidates}
\For{$j = 1$ \KwTo $n$}{
    $d_H = Haus^H(\mathbf{Q}^\mathbf{H},\mathbf{H}_j)$\;
    $\mathcal{C} = \mathcal{C} \cup \{(\mathbf{V}_j, d_H)\}$\;
}
$\mathcal{F} = \{(\mathbf{V}_i, d_H) \in \mathcal{C} \mid d_H \leq d_H^{(c)}\}$, where $d_H^{(c)}$ is the $c$-th smallest $d_H$ in $\mathcal{C}$\;
\tcp{Select top-$k$ results}
$\mathcal{D} = \emptyset$\;
\For{each $(\mathbf{V}_i, d_H) \in \mathcal{F}$}{
    $d_i = Haus(\mathbf{Q}, \mathbf{V}_i)$\;
    $\mathcal{D} = \mathcal{D} \cup \{(\mathbf{V}_i, d_i)\}$\;
}
$\mathbf{R} = \{(\mathbf{V}_i, d_i) \in \mathcal{D} \mid d_i \leq d_i^{(k)}\}$, where $d_i^{(k)}$ is the $k$-th smallest $d_i$ in $\mathcal{D}$\;
\Return $\mathbf{R}$\;
\end{algorithm}

\subsubsection{Search Execution in \NaiveBioVSS}
Binary codes offer computational efficiency by enabling rapid bitwise operations, which are natively supported by modern CPU architectures \cite{mula2018faster}. This results in substantial acceleration compared to conventional floating-point arithmetic.

Algorithm 2 presents \NaiveBioVSS procedure for performing top-$k$ search within vector set databases. The inputs include query set $\mathbf{Q}$, target result count $k$, database $\mathbf{D}$, its associated sparse binary representations $\mathbf{D}^\mathbf{H}$, and the \WTA parameter $L_{wta}$. The output $\mathbf{R}$ comprises top-$k$ vector sets most similar to the query. The algorithm first encodes the query set $\mathbf{Q}$ into its sparse binary form $\mathbf{Q}^\mathbf{H}$ (lines 2-5). To take advantage of CPU-optimized bit-level computation, a Hamming-based Hausdorff distance $Haus^H$ is computed (using $hamming(\mathbf{q}, \mathbf{v})$ \cite{hamming_dist} in place of $dist$ in Definition \ref{def:Hausdorff}) between $\mathbf{Q}^\mathbf{H}$ and each binary code set $\mathbf{H}_j$ in $\mathbf{D}^\mathbf{H}$ (lines 6-8). This step yields a candidate set $\mathcal{F}$ of potentially similar vector sets based on binary code proximity (line 9). In the refinement phase, the algorithm evaluates the exact Hausdorff distance (see Definition \ref{def:Hausdorff}) between $\mathbf{Q}$ and each candidate set $\mathbf{V}_i$ in the original vector space (lines 10-13). It then selects the $k$ closest matches (line 14) and returns the final result set $\mathbf{R}$ (line 15).

\subsection{Theoretical Analysis of Algorithm Correctness}
\label{sec: theory}
This section presents a theoretical analysis of the probabilistic guarantees for result correctness in \NaiveBioVSS. We define a function constraining similarity measure bounds, followed by upper and lower tail probability bounds for similarity comparisons. Through several lemmas, we construct our theoretical framework, culminating in Theorem \ref{theorem: ultimate}, which establishes the relationship between the error rate $\delta$ and $L=L_{wta}$. 

\myparagraph{Assumptions}
\label{sec:assumptions}
Our framework assumes all vectors undergo $L2$ normalization. This allows us to define similarity between vectors $\mathbf{q}$ and $\mathbf{v}$ as their inner product: $sim(\mathbf{q}, \mathbf{v}) = \mathbf{q}^T\mathbf{v}$. We denote the query vector set and another vector set as $\mathbf{Q}$ and $\mathbf{V}$, respectively. To facilitate proof, we define $Sim_{Haus}(\mathbf{Q}, \mathbf{V}) = \min(\min_{\mathbf{q} \in \mathbf{Q}} \max_{\mathbf{v} \in \mathbf{V}} sim(\mathbf{q}, \mathbf{v}), \min_{\mathbf{v} \in \mathbf{V}} \max_{\mathbf{q} \in \mathbf{Q}} sim(\mathbf{q}, \mathbf{v}))$. The similarity metric is equivalent to the Hausdorff distance for $L2$-normalized vectors.

\subsubsection{Bounds on Min-Max Similarity Scores} we begin by introducing upper and lower bounds for a function $\sigma$, which forms the foundation for our subsequent proofs.

\begin{lemma} 
\label{lemma:sigma_bounds}
Consider a matrix $\mathbf{S} = [s_{ij}] \in \mathbb{R}^{m_q \times m}$ containing similarity scores between a query vector set $\mathbf{Q}$ and a target vector set $\mathbf{T}$, where $|\mathbf{Q}| = m_q$ and $|\mathbf{T}| = m$. Define a function $\sigma: \mathbb{R}^{m_q \times m} \to \mathbb{R}$ as
$\sigma(\mathbf{S}) = \min(\min_i \max_j s_{ij}, \min_j \max_i s_{ij})$
Then, $\sigma(\mathbf{S})$ satisfies the following bounds:
$
\min_{i,j} s_{ij} \leq \sigma(\mathbf{S}) \leq \max_{i,j} s_{ij}.
$
\end{lemma}

\begin{proof}
Let $a = \min_i \max_j s_{ij}$ and $b = \min_j \max_i s_{ij}$. Then $\sigma(\mathbf{S}) = \min(a, b)$.

For the lower bound:
$
    \forall i,j: s_{ij} \leq \max_j s_{ij} \Rightarrow \min_{i,j} s_{ij} \leq \min_i \max_j s_{ij} = a.
$
Similarly, $\min_{i,j} s_{ij} \leq b$. Therefore, $\min_{i,j} s_{ij} \leq \min(a, b) = \sigma(\mathbf{S})$.

For the upper bound:
$
    \forall i: \max_j s_{ij} \leq \max_{i,j} s_{ij} \Rightarrow a = \min_i \max_j s_{ij} \leq \max_{i,j} s_{ij}.
$
Similarly, $b \leq \max_{i,j} s_{ij}$. Therefore, $\sigma(\mathbf{S}) = \min(a, b) \leq \max_{i,j} s_{ij}$.

Thus, we have $\min_{i,j} s_{ij} \leq \sigma(\mathbf{S}) \leq \max_{i,j} s_{ij}$.
\end{proof}

\subsubsection{Upper Tail Probability Bound}
we now establish the upper tail probability bound for our relevance score function, which forms the foundation for our subsequent proofs.

\begin{lemma}
\label{lemma:max_similarity_bound}
Consider a similarity matrix $\mathbf{S} \in \mathbb{R}^{m_q \times m}$ between a query vector set and a target vector set. Define the maximum similarity score in $\mathbf{S}$ as $s_{\max} = \max_{i,j} s_{ij}$. Given an estimated similarity matrix $\hat{\mathbf{S}}$ of $\mathbf{S}$ and a threshold $\tau_1 \in (s_{\max}, 1)$, we define $\Delta_1 = \tau_1 - s_{\max}$. Then, the following inequality holds:
\begin{equation*}
\operatorname{Pr}\left[\sigma(\hat{\mathbf{S}}) \geq s_{\max }+\Delta_1\right] \leq m_q m \gamma^L,
\end{equation*}
for $\gamma=\left(\frac{s_{\max }(1-\tau_1)}{\tau_1\left(1-s_{\max }\right)}\right)^{\tau_1}\left(\frac{1-s_{\max }}{1-\tau_1}\right).$
Here, $\sigma(\cdot)$ is the operator defined in Lemma \ref{lemma:sigma_bounds}, and $L \in \mathbb{Z}^+$ represents the number of hash functions.
\end{lemma}

\begin{proof}
Applying a generic Chernoff bound to $\sigma(\hat{\mathbf{S}})$ yields the following bounds for any $t>0$:
\begin{align*}
\operatorname{Pr}[\sigma(\hat{\mathbf{S}}) \geq \tau_1] &=  \operatorname{Pr}\left[e^{t \sigma(\hat{\mathbf{S}})} \geq e^{t \tau_1}\right]
\leq \frac{\mathbb{E}\left[e^{t \sigma(\hat{\mathbf{S}})}\right]}{e^{t \tau_1}}.
\end{align*}

Given an \LSH family (see Definition \ref{def: FlyHash}), the estimated maximum similarity $\hat{s}_{\max}$ follows a scaled binomial distribution with parameters $s_{\max}$ and $L^{-1}$, i.e., $\hat{s}_{\max } \sim L^{-1} \mathcal{B}\left(s_{\max }, L\right)$. Substituting the binomial moment generating function into the expression and using Lemma \ref{lemma:sigma_bounds}, we can bound the numerator:
\begin{align*}
\mathbb{E}\left[e^{t \sigma(\hat{\mathbf{S}})}\right] &\leq \mathbb{E}\left[e^{t \max_{1 \leq i \leq m_q, 1 \leq j \leq m} \hat{S}_{ij}}\right] \\
&\leq m_q m \mathbb{E}\left[e^{t \hat{s}_{\max }}\right] 
= m_q m \left(1-s_{\max }+s_{\max } e^{\frac{t}{L}}\right)^L.
\end{align*}

Combining the Chernoff bound and the numerator bound yields:
$
\operatorname{Pr}[\sigma(\hat{\mathbf{S}}) \geq \tau_1] \leq m_q m e^{-t \tau_1}\left(1-s_{\max }+s_{\max } e^{\frac{t}{L}}\right)^L.
$

To find the tightest upper tail probability bound, we can determine the infimum by setting the derivative of the bound for $t$ equal to zero. Then, we get:
$
t^{\star}=L \ln \left(\frac{\tau_1\left(1-s_{\max }\right)}{s_{\max }(1-\tau_1)}\right).
$

Since $\tau_1 \in (s_{\max}, 1)$, the numerator of the fraction inside the logarithm is greater than the denominator, ensuring that $t^{\star} > 0$. This result allows us to obtain the tightest bound.

Substituting $t=t^{\star}$ into the bound, we obtain:
$$
\operatorname{Pr}[\sigma(\hat{\mathbf{s}}) \geq \tau_1] \leq m_q m \left(\left(\frac{\tau_1\left(1-s_{\max }\right)}{s_{\max }(1-\tau_1)}\right)^{-\tau_1}\left(\frac{1-s_{\max }}{1-\tau_1}\right)\right)^L.
$$

Thus we have:
$
\gamma=\left(\frac{s_{\max }(1-\tau_1)}{\tau_1\left(1-s_{\max }\right)}\right)^{\tau_1}\left(\frac{1-s_{\max }}{1-\tau_1}\right).
$
\end{proof}

\subsubsection{Lower Tail Probability Bound}
we now establish the lower tail probability bound for our relevance score function, which forms the foundation for our subsequent proofs.

\begin{lemma}
\label{lemma:min_similarity_bound}
Consider a similarity matrix $\mathbf{S} \in \mathbb{R}^{m_q \times m}$ between a query vector set and a target vector set. Define the minimum similarity score in $\mathbf{S}$ as $s_{\min} = \min_{i,j} s_{ij}$. Given an estimated similarity matrix $\hat{\mathbf{S}}$ of $\mathbf{S}$ and a threshold $\tau_2 \in (0, s_{\min})$, we define $\Delta_2 = s_{\min}-\tau_2$. Then, the following inequality holds:
\begin{equation*}
\operatorname{Pr}\left[\sigma(\hat{\mathbf{s}}) \leq s_{\min }-\Delta_2 \right] \leq m_q m \xi^L, 
\end{equation*}
for $\gamma=\left(\frac{s_{\min }(1-\tau_2)}{\tau_2\left(1-s_{\min }\right)}\right)^{\tau_2}\left(\frac{1-s_{\min }}{1-\tau_2}\right).$
Here, $\sigma(\cdot)$ is the operator defined in Lemma \ref{lemma:sigma_bounds}, and $L \in \mathbb{Z}^+$ represents the number of hash functions.
\end{lemma}

\begin{proof}
Applying a generic Chernoff bound to $\sigma(\hat{\mathbf{S}})$ yields the following inequality for any $t<0$:
\begin{align*}
\operatorname{Pr}[\sigma(\hat{\mathbf{S}}) \leq \tau_2] &= \operatorname{Pr}\left[e^{t \sigma(\hat{\mathbf{S}})} \geq e^{t \tau_2}\right] 
\leq \frac{\mathbb{E}\left[e^{t \sigma(\hat{\mathbf{S}})}\right]}{e^{t \tau_2}}.
\end{align*}

Given an \LSH family (see Definition \ref{def: FlyHash}), the estimated minimum similarity $\hat{s}_{\min}$ follows a scaled binomial distribution with parameters $s_{\min}$ and $L^{-1}$, i.e., $\hat{s}_{\min } \sim L^{-1} \mathcal{B}\left(s_{\min }, L\right)$. Substituting the binomial moment generating function into the expression and using Lemma \ref{lemma:sigma_bounds}, we can bound the numerator:
\begin{align*}
\mathbb{E}\left[e^{t \sigma(\hat{\mathbf{S}})}\right] &\leq \mathbb{E}\left[e^{t \min_{1 \leq i \leq m_q, 1 \leq j \leq m} \hat{S}_{ij}}\right] \\
&\leq m_q m \mathbb{E}\left[e^{t \hat{s}_{\min }}\right] 
= m_q m \left(1-s_{\min }+s_{\min } e^{\frac{t}{L}}\right)^L.
\end{align*}

Combining the Chernoff bound and the numerator bound yields:
$
\operatorname{Pr}[\sigma(\hat{\mathbf{S}}) \leq \tau_2] \leq m_q m e^{-t \tau_2}\left(1-s_{\min }+s_{\min } e^{\frac{t}{L}}\right)^L.
$

To find the tightest upper bound, we can determine the infimum by setting the derivative of the upper bound with respect to $t$ equal to zero. Then, we obtain:
$
t^{\star}=L \ln \left(\frac{\tau_2\left(1-s_{\min }\right)}{s_{\min }(1-\tau_2)}\right).
$

Since $\tau_2 \in (0, s_{\min})$, the numerator of the fraction inside the logarithm is greater than the denominator, ensuring that $t^{\star} < 0$. This result allows us to obtain the tightest upper bound.

Similar to Lemma \ref{lemma:max_similarity_bound}, we get:
$
\xi=\left(\frac{s_{\min }(1-\tau_2)}{\tau_2\left(1-s_{\min }\right)}\right)^{\tau_2}\left(\frac{1-s_{\min }}{1-\tau_2}\right).
$
\end{proof}

\subsubsection{Probabilistic Correctness for Search Results}
this section presents a theorem establishing the probabilistic guarantees to the approximate top-$k$ vector set search problem.  

\begin{theorem}
\label{theorem: ultimate}
Let $\mathbf{V}_{\alpha}^{\star}$ be one of the top-$k$ vector sets that minimize the Hausdorff distance $Haus(\mathbf{Q}, \mathbf{V})$, i.e., $\mathbf{V}_{\alpha}^{\star} \in \operatorname{argmin}^k_{\mathbf{V} \in \mathbf{D}} Haus(\mathbf{Q}, \mathbf{V})$, and $\mathbf{V}_{\beta}$ be any other vector set in the database $\mathbf{D}$. By applying Lemma \ref{lemma:sigma_bounds}, we obtain the following upper and lower bounds for $Haus(\mathbf{Q}, \mathbf{V}_{\alpha}^{\star})$ and $Haus(\mathbf{Q}, \mathbf{V}_{\beta})$:
$$
B_{\alpha}^{\star}= 1 - \min_{i,j}{S_{ij}^{\alpha}} = 1- s_{\alpha, \min}, \quad
B_{\beta}= 1 - \max_{i,j}{S_{ij}^{\beta}}=1 - s_{\beta, \max},
$$
where $S_{ij}^{\alpha}$ and $S_{ij}^{\beta}$ denote the element-wise similarity matrices between each vector in the query set $\mathbf{Q}$ and each vector in $\mathbf{V}_{\alpha}^{\star}$ and $\mathbf{V}_{\beta}$.

Let $B^{\prime} = \min_{\beta} B_{\beta}$ be the minimum value of $B_{\beta}$ over any vector set $\mathbf{V}_{\beta}$ not in the top-$k$ results, and let $B^{\star} = \max_{\alpha}B_{\alpha}^{\star}$ be the maximum value of $B_{\alpha}^{\star}$ over any top-$k$ vector set $\mathbf{V}_{\alpha}^{\star}$. Define $\Delta_1 + \Delta_2$ as follows:
$$
B_{\beta} - B_{\alpha}^{\star}  =  s_{\alpha, \min} - s_{\beta, \max}  \geq  
B^{\prime} - B^{\star}= 2(\Delta_1 + \Delta_2).
$$
If $\Delta_1>0$ and $\Delta_2>0$, a hash structure with $L=O\left(\log \left(\frac{n m_q m}{\delta}\right)\right)$ solves the approximate top-$k$ vector set search problem (Definition \ref{ApproximateSearch}) with probability at least $1-\delta$, where $n=|\mathbf{D}|$, $m_q=|\mathbf{Q}|$, $m=|\mathbf{V}_i|$ (Assume \( m \) is invariant across all \(\mathbf{V}_i\)), and $\delta$ is the failure probability.
\end{theorem}

\begin{proof}
We aim to determine the value of $L$ that satisfies both tail bounds and combine them to obtain the final $L$.

\myparagraph{Upper Tail Probability Bound} For $\mathbf{V}_{\beta} \notin \operatorname{argmin}^k_{\mathbf{V} \in \mathbf{D}} Haus(\mathbf{Q}, \mathbf{V})$, Lemma \ref{lemma:max_similarity_bound} yields:
$$
\operatorname{Pr}\left[\sigma(\hat{\mathbf{S}}_{\beta}) \geq s_{\beta, \max}+\Delta_1\right] \leq m_q m \gamma_{\beta}^L,
$$

where $\gamma_{\beta}=\left(\frac{s_{\beta, \max}(1-\tau_1)}{\tau_1\left(1-s_{\beta, \max}\right)}\right)^{\tau_1}\left(\frac{1-s_{\beta, \max}}{1-\tau_1}\right)$. 

To simplify the analysis, we consider $\gamma_{\max} = \max_{\mathbf{V}_{\beta}} \gamma_{\beta}$, allowing us to express all bounds using the same $\gamma_{\max}$. To ensure that the probability of the bound holding for all $n-k$ non-top-$k$ sets is $\frac{\delta}{2}$, we choose $L$ such that $L\geq \log \frac{2(n-k) m_q m}{\delta} \left(\log \frac{1}{\gamma_{\max }}\right)^{-1}$. Consequently,
\begin{align*}
\Pr\left[\sigma(\hat{\mathbf{S}}_{\beta}) \geq s_{\beta, \max}+\Delta_1\right] 
& \leq m_q m \gamma_{\beta}^L \leq m_q m \gamma_{\max}^L \\
& \leq m_q m \gamma_{\max}^{\log \frac{2(n-k) m_q m}{\delta} \left(\log \frac{1}{\gamma_{\max }}\right)^{-1}}\\
&=\frac{\delta}{2(n-k)}.
\end{align*}

\myparagraph{Lower Tail Probability Bound} For $\mathbf{V}_{\alpha}^{\star} \in \operatorname{argmin}^k_{\mathbf{V} \in \mathbf{D}} Haus(\mathbf{Q}, \mathbf{V})$, Lemma \ref{lemma:min_similarity_bound} yields:
$$
\operatorname{Pr}\left[\sigma(\hat{\mathbf{S}}_{\alpha}) \leq s_{\alpha, \min }-\Delta_2 \right] \leq m_q m \xi_{\alpha}^L,
$$
where $\xi_{\alpha}= \left(\frac{s_{\alpha, \min}(1-\tau_2)}{\tau_2\left(1-s_{\alpha, \min}\right)}\right)^{\tau_2}\left(\frac{1-s_{\alpha, \min}}{1-\tau_2}\right)$.

Following a similar approach as for the upper bound, we ensure that the probability of the bound holding for all top-$k$ sets is $\frac{\delta}{2}$ by selecting $L$ such that $L\geq \log \frac{2k m_q m}{\delta} \left(\log \frac{1}{\xi_{\max }}\right)^{-1}$, where $\xi_{\max} = \max_{\alpha} \xi_{\alpha}$. As a result,
\begin{align*}
\Pr\left[\sigma(\hat{\mathbf{S}}_{\alpha}) \leq s_{\alpha, \min}-\Delta_2\right] 
&\leq m_q m \xi_{\alpha}^L 
\leq m_q m \xi_{\max}^L \\
&\leq m_q m \xi_{\max}^{\log \frac{2k m_q m}{\delta} \left(\log \frac{1}{\xi_{\max }}\right)^{-1}} \\
&=\frac{\delta}{2k}.
\end{align*}

\myparagraph{Combining the Bounds}
Let 
$
L=\max \left( \frac{\log \frac{2(n-k) m_q m}{\delta}}{\log \left(\frac{1}{\gamma_{\max }}\right)} , \frac{\log \frac{2k m_q m}{\delta}}{\log \left(\frac{1}{\xi_{\max }}\right)}\right),
$
we calculate the final $L$ by combining the upper and lower tail probability bounds. Let $\vmathbb{1}$ be an indicator random variable that equals 1 when all the upper and lower tail
probability bounds are satisfied, and 0 otherwise. The probability of solving the approximate top-$k$
vector set problem is equivalent to the probability that:

\begin{footnotesize}
\begin{equation*}
\begin{aligned}
& \operatorname{Pr}\left(\forall_{\alpha, \beta}\left(\hat{H}\left(\mathbf{Q}, \mathbf{V}_{\alpha}^{\star} \right)-\hat{H}\left(\mathbf{Q}, \mathbf{V}_{\beta}\right)<0\right)\right) \\
&=\operatorname{Pr}\left(\forall_{\alpha, \beta}\left(\min\left(\min_{\mathbf{q} \in \mathbf{Q}} \max_{\mathbf{v} \in \mathbf{V}_{\alpha}^{\star}}(sim(\mathbf{q}, \mathbf{v})), \min_{\mathbf{v} \in \mathbf{V}_{\alpha}^{\star}} \max_{\mathbf{q} \in \mathbf{Q}}(sim(\mathbf{q}, \mathbf{v}))\right) \right.\right. \\
& \left.\left. \quad - \min\left(\min_{\mathbf{q} \in \mathbf{Q}} \max_{\mathbf{v} \in \mathbf{V}_{\beta}}(sim(\mathbf{q}, \mathbf{v})), \min_{\mathbf{v} \in \mathbf{V}_{\beta}} \max_{\mathbf{q} \in \mathbf{Q}}(sim(\mathbf{q}, \mathbf{v}))\right) > 0\right)\right) \hfill \text{} \\
&=\operatorname{Pr}\left(\forall_{\alpha, \beta}\left(\min\left(\min_{1 \leq i \leq m_q} \max_{1 \leq j \leq m} S_{ij}^{\alpha}, \min_{1 \leq j \leq m} \max_{1 \leq i \leq m_q} S_{ij}^{\alpha}\right) \right.\right. \\
& \left.\left. \quad - \min\left(\min_{1 \leq i \leq m_q} \max_{1 \leq j \leq m} S_{ij}^{\beta}, \min_{1 \leq j \leq m} \max_{1 \leq i \leq m_q} S_{ij}^{\beta}\right) > 0\right)\right)  \\
& \ge \operatorname{Pr}\left(\forall_{\alpha, \beta}\left( \min\left(\min_{1 \leq i \leq m_q} \max_{1 \leq j \leq m} S_{ij}^{\alpha}, \min_{1 \leq j \leq m} \max_{1 \leq i \leq m_q} S_{ij}^{\alpha}\right) \right.\right. \hfill \text{} \\
& \left.\left. \quad - \min\left(\min_{1 \leq i \leq m_q} \max_{1 \leq j \leq m} S_{ij}^{\beta}, \min_{1 \leq j \leq m} \max_{1 \leq i \leq m_q} S_{ij}^{\beta}\right) > 0
\mid \vmathbb{1}=1 \right)\right) \operatorname{Pr}(\vmathbb{1}=1) \\
& \ge \operatorname{Pr}\left(\forall_{\alpha, \beta}\left(\min{S_{ij}^{\alpha}} - \Delta_2 - (\max{S_{ij}^{\beta}} + \Delta_1) >0
\mid \vmathbb{1}=1 \right)\right) \operatorname{Pr}(\vmathbb{1}=1) \hfill \text{}\\
& = \operatorname{Pr}\left(\forall_{\alpha, \beta}\left(\min{S_{ij}^{\alpha}} -\max{S_{ij}^{\beta}} > \Delta_1 + \Delta_2
\mid \vmathbb{1}=1 \right)\right) \operatorname{Pr}(\vmathbb{1}=1) \\
& = \operatorname{Pr}\left(\forall_{\alpha, \beta}\left( B_{\beta} - B_{\alpha}^{\star}>\Delta_1 + \Delta_2
\mid \vmathbb{1}=1 \right)\right) \operatorname{Pr}(\vmathbb{1}=1) \hfill \text{}\\
& \ge \operatorname{Pr}\left(\forall_{\alpha, \beta}\left(  B^{\prime} - B^{\star}>\Delta_1 + \Delta_2
\mid \vmathbb{1}=1 \right)\right) \operatorname{Pr}(\vmathbb{1}=1) \hfill \qquad \text{}\\
& \ge \operatorname{Pr}\left(\forall_{\alpha, \beta}\left( 2(\Delta_1 + \Delta_2) > \Delta_1 + \Delta_2
\mid \vmathbb{1}=1 \right)\right) \operatorname{Pr}(\vmathbb{1}=1) \hfill \quad \text{}\\
& = 1 *\operatorname{Pr}(\vmathbb{1}=1) \ge 1 -  ( (n-k) *\frac{\delta}{2(n-k)}  + \frac{\delta}{2k} * k) \qquad\qquad \text{} \\
& = 1 -  \delta.   
\end{aligned}
\end{equation*}
\end{footnotesize}

Therefore, with this choice of $L$, the algorithm effectively solves the approximate top-$k$ vector set search problem. By eliminating the data-dependent correlation terms $\gamma_{\max }$ and $\xi_{\max }$, and considering the practical scenario where $n \gg k$, then $L = O\left(\log \left(\frac{n m_q m}{\delta}\right)\right)$.
\end{proof}

\subsection{Performance Discussion}
The time complexity of \NaiveBioVSS is $O(nm^2L/w)$, where $n$ is the number of vector sets, $m$ is the number of vectors per set, $L$ is the binary vector length, and $w$ is the machine word size (typically 32 or 64 bits). The factor $L/w$ represents the number of machine words needed to store each binary vector, allowing for parallel processing of $w$ bits. This results in improved efficiency compared to real-number operations.

Despite this optimization, \NaiveBioVSS still requires an exhaustive scan. To address this limitation, we propose enhancing \NaiveBioVSS with a filter. Next, we will detail this method.

\section{Enhancing BioVSS via Cascade Filter}
\label{Sec_BioVSS}
We introduce \BioVSS, an improved approach that overcomes the linear scan bottleneck of \NaiveBioVSS by incorporating a Bio-Inspired Dual-Layer Cascade Filter (\BDLCF). The core idea behind \BioVSS lies in utilizing the structural parallels between \BioHash and Bloom filters, both of which rely on sparse representations for encoding vectors and managing set membership. In the following, we elaborate on the construction of the filter and the procedure for executing the search.

\subsection{Filter Construction} 
The bio-inspired dual-layer cascade filter serves as the core mechanism enabling efficient search within \BioVSS. \BDLCF is composed of two primary elements: an inverted index built upon count Bloom filters, and vector set sketches constructed using binary Bloom filters. The inverted index utilizes count Bloom filters to generate inverted lists, where vector sets are organized according to their respective count values. For the second-level filtering, the vector set sketches make use of binary Bloom filters, allowing for efficient candidate scanning based on Hamming distance. This two-tiered filtering strategy in \BDLCF effectively boosts search performance by progressively narrowing down the candidate pool through sequential refinement.

\begin{figure}
    \centering
    \captionsetup{aboveskip=2pt}
    \includegraphics[width=0.8\textwidth]{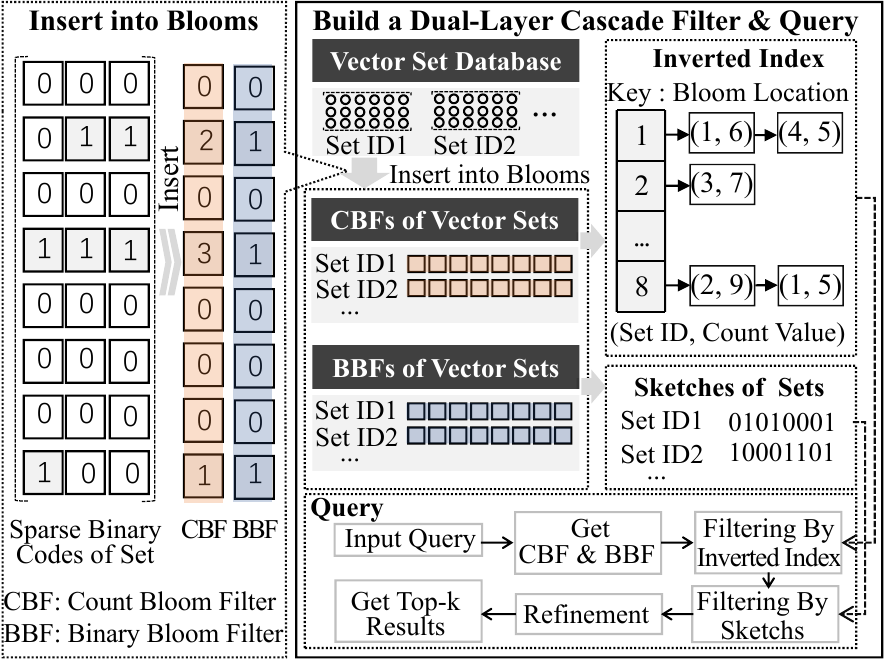}
    \caption{Filter Construction and Query Execution in \BioVSS}
    \label{fig:framework}
    % \addvspace{-10pt}
\end{figure}

\subsubsection{Inverted Index Based on Count Bloom Filter} 
The inverted index constitutes the first layer of \BDLCF. It improves search efficiency by indexing vector sets in an inverted manner, thereby reducing the number of candidate items.

To overcome the linear scanning bottleneck present in Algorithm \ref{alg:hashing}, we utilize count Bloom filters to manage vector sets, which serve as the foundation for building the inverted index.

\begin{definition1}[\textbf{Count Bloom Filter} \cite{bonomi2006improved}] Let $\mathbf{V} = \{\mathbf{v}_1, \mathbf{v}_2, ..., \mathbf{v}_m\}$ be a vector set. A \textbf{count Bloom filter} $\mathbf{C} = (c_1, c_2, ..., c_b)$ for $\mathbf{V}$ is an array of $b$ counters, where each counter $c_i$ represents the sum of the $i$-th bits of the sparse binary codes of vectors in $\mathbf{V}$. Formally, $c_i = \sum_{j=1}^m \mathcal{H}(\mathbf{v}_j)_i$, where $\mathcal{H}(\mathbf{v}_j)_i$ is the $i$-th bit of the binary code of $\mathbf{v}_j$. \end{definition1}

The count Bloom filter encodes each binary code's frequency across all vectors in the set. Next, we define the inverted index, which is built upon the count Bloom filters.

\begin{definition1}[\textbf{Inverted Index Based on Count Bloom Filter}] Let $\mathbf{D} = \{\mathbf{V}_1, \mathbf{V}_2, ..., \mathbf{V}_n\}$ be a database of vector sets. A \textbf{inverted index based on count bloom filter} $\mathbf{I}$ is a data structure that maps each bit position to a sorted list of tuples. For each position $i \in \{1, 2, ..., b\}$, $\mathbf{I}[i]$ contains a list of tuples $(j, c_i^{(j)})$, where $j$ is the index of the vector set $\mathbf{V}_j \in \mathbf{D}$, and $c_i^{(j)}$ is the $i$-th count value in the count Bloom filter for $\mathbf{V}_j$. The list is sorted in descending order based on $c_i^{(j)}$ values. \end{definition1}

\setlength{\textfloatsep}{0cm}
\begin{algorithm}[t]
\small
\caption{\texttt{Gen\_Count\_Bloom\_Filter}($\mathbf{D}, L_{wta}, b$)}
\label{alg:gen_count_bloom_filter}
\KwIn{$\mathbf{D} = \{\mathbf{V}_i\}_{i=1}^n$: vector set database, $L_{wta}$: \# of \WTA, $b$: the length of bloom filters}
\KwOut{$\mathcal{C}$: count Bloom filters}

\tcp{Generate binary codes: Algorithm \ref{alg:hashing}}
$\mathbf{D}^\mathbf{H} = \texttt{Gen\_Binary\_Codes}(\mathbf{D}, L_{wta})$\;

\tcp{Construct Count Bloom Filters}
$\mathcal{C} = \{\mathbf{C}^{(j)}\}_{j=1}^n$, where $\mathbf{C}^{(j)} \in \mathbb{N}^b$\;
\For{$j = 1$ \KwTo $n$}{
    $\mathbf{C}^{(j)} = \sum_{\mathbf{h} \in \mathbf{H}_j} \mathbf{h}$, where $\mathbf{H}_j \in \mathbf{D}^\mathbf{H}$\;
}

\Return $\mathcal{C}$\;
\end{algorithm}
\setlength{\textfloatsep}{0cm}

\setlength{\textfloatsep}{0cm}
\begin{algorithm}[]
\small
\caption{\texttt{Build\_Inverted\_Index}($\mathbf{D}, L_{wta}, b$)}
\label{alg:bcbii_construction}
\KwIn{$\mathbf{D} = \{\mathbf{V}_i\}_{i=1}^n$: vector set database, $L_{wta}$: \# of \WTA, $b$: the length of bloom filters}
\KwOut{$\mathbf{I}$: inverted index}

\tcp{Generate Count Bloom Filters: Algorithm \ref{alg:gen_count_bloom_filter}}
$\mathcal{C} = \texttt{Gen\_Count\_Bloom\_Filter}(\mathbf{D}, L_{wta}, b)$\;

\tcp{Build Inverted Index}
Initialize $\mathbf{I} = \{\mathbf{I}_i\}_{i=1}^b$, where $\mathbf{I}_i = \emptyset$ for $i = 1, \ldots, b$\;
\For{$j = 1$ \KwTo $n$}{
    \For{$i = 1$ \KwTo $b$}{
        $c_i^{(j)} = \mathbf{C}^{(j)}[i]$\;
        $\mathbf{I}_i = \mathbf{I}_i \cup \{(j, c_i^{(j)})\}$\;
    }
}
\For{$i = 1$ \KwTo $b$}{
    Sort $\mathbf{I}_i$ in descending order by $c_i^{(j)}$\;
}

\Return $\mathbf{I}$\;
\end{algorithm}
\setlength{\textfloatsep}{0cm}

The inverted index plays a critical role in enabling efficient vector set search. As illustrated in Figure \ref{fig:framework}, each vector set is first inserted into a count Bloom filter. Inverted lists are then constructed and ranked in descending order according to the count values at each filter position. This strategy leverages the intuition that higher count values imply a greater probability of collisions at the same positions, suggesting stronger similarity among vector sets.

Algorithm \ref{alg:gen_count_bloom_filter} outlines the procedure for generating count Bloom filters across all vector sets. The inputs include a database $\mathbf{D}$ containing $n$ vector sets, the winner-takes-all parameter $L_{wta}$, and the desired length $b$ of the Bloom filters. Initially, sparse binary codes for the dataset are generated using the \texttt{Gen\_Binary\_Codes} procedure defined in Algorithm \ref{alg:hashing} (line 1). Subsequently, a count Bloom filter $\mathcal{C}$ is computed for each vector set by aggregating the binary codes of its constituent vectors. Specifically, each count Bloom filter $\mathbf{C}^{(j)}$ is produced by summing the binary vectors in $\mathbf{H}_j$ (lines 2–4), followed by returning the complete set of filters (line 5).

Algorithm \ref{alg:bcbii_construction} presents the method for constructing the inverted index. This algorithm takes the same input as Algorithm \ref{alg:gen_count_bloom_filter}, and outputs the inverted index $\mathbf{I}$. First, count Bloom filters for all entries in the dataset are generated using Algorithm \ref{alg:gen_count_bloom_filter} (line 1). For each position $i$ in the filters, an inverted list $\mathbf{I}_i$ is created, storing pairs of set indices and their corresponding count values (lines 2–6). Finally, each list is sorted in descending order based on these count values (lines 7–8). This index structure effectively narrows down the search space by filtering candidates early in the process.

\subsubsection{Vector Set Sketches Based on Binary Bloom Filter}
Vector set sketches constitute the second layer of \BDLCF, offering a unified binary representation for each vector set. These sketches are generated using binary Bloom filters, where a bitwise OR operation is applied across the binary codes of all vectors within a set.

Due to their binary structure, these sketches allow efficient similarity estimation via Hamming distance, utilizing fast XOR and popcount instructions supported by modern CPU architectures \cite{mula2018faster}. This design significantly reduces computational overhead by eliminating the need for costly aggregation steps inherent in Hausdorff distance calculations.

To better understand the operation of the building set sketches, we first present the binary Bloom filter:

\begin{definition1}[\textbf{Binary Bloom Filter} \cite{bonomi2006improved}]
\label{def_bloom_hash_code}
Let $\mathbf{V} = \{\mathbf{v}_1, \mathbf{v}_2, ..., \mathbf{v}_m\}$ be a set of vectors. The \textbf{binary Bloom filter} $\mathbf{B}$ for $\mathbf{V}$ is a binary array of length $b$, obtained by performing a bitwise OR operation on the binary codes all vectors in $\mathbf{V}$. Formally, $\mathbf{B} = \mathcal{H}(\mathbf{v}_1) \lor \mathcal{H}(\mathbf{v}_2) \lor ... \lor \mathcal{H}(\mathbf{v}_m)$, where $\mathcal{H}(\mathbf{v}_i)$ is the binary vector of $\mathbf{v}_i$ and $\lor$ denotes the bitwise OR operation.
\end{definition1}

As illustrated in Figure \ref{fig:framework}, each vector set is first inserted into a binary Bloom filter. This binary Bloom filter is then converted into a corresponding vector set sketch.

\setlength{\textfloatsep}{0cm}
\begin{algorithm}[t]
\small
\caption{\texttt{Build\_Set\_Sketches}($\mathbf{D}, L_{wta}, b$)}
\label{alg:bbhc_construction}
\KwIn{$\mathbf{D} = \{\mathbf{V}_i\}_{i=1}^n$: vector set database, $L_{wta}$: \# of \WTA, $b$: the length of bloom filters}
\KwOut{$\mathcal{S} = \{\mathbf{S}^{(i)}\}_{i=1}^n$: set sketches}

\tcp{Generate Binary Codes: Algorithm \ref{alg:hashing}}
$\mathbf{D}^\mathbf{H} = \texttt{Gen\_Binary\_Codes}(\mathbf{D}, L_{wta})$;

\tcp{Construct Binary Bloom Filters}
$\mathcal{B} = \{\mathbf{B}^{(i)}\}_{i=1}^n$, where $\mathbf{B}^{(i)} \in \{0,1\}^b$\;
\For{$i = 1$ \KwTo $n$}{
    $\mathbf{B}^{(i)} = \bigvee_{\mathbf{h} \in \mathbf{H}_i} \mathbf{h}$, where $\mathbf{H}_i \in \mathbf{D}^\mathbf{H}$\tcp*{$\bigvee$ is $OR$}
}
\tcp{Build Set Sketches}
$\mathcal{S} = \{\mathbf{S}^{(i)}\}_{i=1}^n$, where $\mathbf{S}^{(i)} = \mathbf{B}^{(i)}$\;

\Return $\mathcal{S}$\;
\end{algorithm}
\setlength{\textfloatsep}{0cm}

Algorithm \ref{alg:bbhc_construction} details the procedure for constructing set sketches from a vector set database. The input includes a database $\mathbf{D}$ containing $n$ vector sets, the winner-takes-all parameter $L_{wta}$, and the binary code length $b$. The output is a set of sketches $\mathcal{S}$, providing compact encodings for each vector set. The algorithm starts by invoking the \texttt{Gen\_Binary\_Codes} function from Algorithm \ref{alg:hashing} to produce sparse binary codes for the entire database (line 1). It then constructs binary Bloom filters $\mathcal{B}$ for all vector sets (lines 2-4), where each binary Bloom filter $\mathbf{B}^{(i)}$ is obtained via a bitwise OR operation ($\bigvee$) over the binary codes of all vectors in the associated set $\mathbf{H}_i$. In the final step (line 5), the binary Bloom filters themselves are used directly as vector set sketches $\mathcal{S}$, which are then returned.

\setlength{\textfloatsep}{0cm}
\begin{algorithm}[!ht]
\small
\caption{\texttt{BioVSS++\_Topk\_Search}($\mathbf{Q}, k, \mathbf{D}, \mathbf{I}, \mathcal{S}, \Theta$)}
\label{alg:vector_set_search}
\KwIn{
    $\mathbf{Q}$: query vector set,
    $k$: \# of results,
    $\mathbf{D} = \{\mathbf{V}_i\}_{i=1}^n$: vector set database,
    $\mathbf{I}$: inverted index,
    $\mathcal{S} = \{\mathbf{S}^{(i)}\}_{i=1}^n$: set sketches,
    $\Theta = \{A: \text{access number of lists}, M: \text{minimum count}, c: \text{the size of candidate set}, T: \text{\# of candidates}, L_{wta}: \text{\# of \WTA}, b: \text{the length of bloom filters}\}$
}
\KwOut{$\mathcal{R}$: top-$k$ vector sets}

\tcp{Get Query Count Bloom Filter and Sketch: Algorithm \ref{alg:gen_count_bloom_filter} and \ref{alg:bbhc_construction}}  
$\mathbf{C}_Q = \texttt{Gen\_Count\_Bloom\_Filter}(\mathbf{Q})$\;
$\mathbf{S}_Q = \texttt{Build\_Set\_Sketch}(\mathbf{Q},L_{wta}, b)$\;

\tcp{Filtering by Inverted Index}
$\pi = \texttt{ArgsortDescending}(\mathbf{C}_Q)$ \;
$\mathcal{P} = \{\pi(i) : i \in [1,A]\}$\;
$\mathcal{F}_1 = \emptyset$\;
\For{$p \in \mathcal{P}$}{
    \For{$(i, c_p^{(i)}) \in \mathbf{I}[p]$}{
        \If{$c_p^{(i)} \geq M$}{
            $\mathcal{F}_1 = \mathcal{F}_1 \cup \{i\}$\;
        }
    }
}

\tcp{Filtering by Sketches}
Initialize $\mathcal{G} \gets$ empty max-heap  with capacity $c$\;
\For{$i \in \mathcal{F}_1$}{
    $d = \texttt{Hamming}(\mathbf{S}_Q, \mathbf{S}^{(i)})$\;
    \If{$|\mathcal{G}| < T$}{
        $\mathcal{G}.\texttt{push}((d, i))$\;
    }
    \ElseIf{$d < \mathcal{G}.\texttt{top}()[0]$}{
        $\mathcal{G}.\texttt{pop}()$\;
        $\mathcal{G}.\texttt{push}((d, i))$\;
    }
}
$\mathcal{F}_2 = \{i | (\_, i) \in \mathcal{G}\}$\;

\tcp{Select top-$k$ results}
$\mathcal{D} = \emptyset$\;
\For{each $(i, d_H) \in \mathcal{F}$}{
    $d_i = Haus(\mathbf{Q}, \mathbf{V}_i)$\;
    $\mathcal{D} = \mathcal{D} \cup \{(\mathbf{V}_i, d_i)\}$\;
}
$\mathbf{R} = \{(\mathbf{V}_i, d_i) \in \mathcal{D} \mid d_i \leq d_i^{(k)}\}$, where $d_i^{(k)}$ is the $k$-th smallest $d_i$ in $\mathcal{D}$\;

\Return $\mathbf{R}$\;
\end{algorithm}
\setlength{\textfloatsep}{0cm}

\subsection{Search Execution in \BioVSS}
This section presents the search execution process in \BioVSS. Leveraging the inverted index and set sketches, \BDLCF can efficiently eliminate unrelated vector sets.

Our search execution approach is based on a dual-layer filtering mechanism. At the first layer, we utilize the inverted index to rapidly remove a substantial number of dissimilar vector sets. At the second layer, we apply the vector set sketches for additional filtering. This addresses the drawbacks of linear scanning in Algorithm \ref{alg:hashing}, minimizing the computational overhead.

Algorithm \ref{alg:vector_set_search} describes \BioVSS top-$k$ query process, which utilizes a two-stage filter \BDLCF for efficient vector set search. The input comprises a query set $\mathbf{Q}$, the desired number of results $k$, the database $\mathbf{D}$, a pre-computed inverted index $\mathbf{I}$, set sketches $\mathcal{S}$, and parameters $\Theta$. The output $\mathbf{R}$ holds the top-$k$ most similar vector sets. The algorithm starts by constructing a count Bloom filter $\mathbf{C}_Q$ and a set sketch $\mathbf{S}_Q$ for the query set (lines 1-2). \BDLCF subsequently applies two filtering stages: 1) inverted index filtering (lines 3-9): This stage utilizes the query's count Bloom filter to locate potential candidates. It chooses the top-$A$ positions with the highest counts in $\mathbf{C}_Q$ and extracts vector sets from the inverted index $\mathbf{I}$ that possess counts exceeding a threshold $M$ at these positions. This generates an initial candidate set $\mathcal{F}_1$. 2) sketch-based filtering (lines 10-18): This stage refines the candidates through set sketches. It calculates the Hamming distance between the query sketch $\mathbf{S}_Q$ and each candidate's sketch $\mathbf{S}^{(i)}$, preserving a max-heap of the $T$ closest candidates. This yields a further refined candidate set $\mathcal{F}_2$. Subsequently, the algorithm evaluates the actual Hausdorff distance between $\mathbf{Q}$ and each remaining candidate in $\mathcal{F}_2$ (lines 19-22), determining the top-$k$ results based on these distances (line 23) and then returns it (line 24).

\subsection{Theoretical Analysis of Dual-Layer Filtering Mechanism}
\label{sec: Filtering Mechanism}
\BDLCF framework demonstrates potential compatibility with diverse set-based distance metrics. This versatility stems from its filter structure being independent of specific distance measurements. This dual-layer approach achieves efficiency through progressive refinement: the count Bloom filter-based inverted index rapidly reduces the search space, while the binary Bloom filter-based sketches enable similarity assessment of the remaining candidates. To establish the theoretical foundation, we introduce the concept of set connectivity and analyze its relationship with filter collision patterns.

\begin{definition1}[\textbf{Set Connectivity}]
\label{def:set_connectivity}
For two vector sets $\mathbf{Q}$ and $\mathbf{V}$, their set connectivity is defined as:
$$
Conn(\mathbf{Q},\mathbf{V}) = \sum_{\mathbf{q} \in \mathbf{Q}} \sum_{\mathbf{v} \in \mathbf{V}} sim(\mathbf{q}, \mathbf{v}),
$$
where $sim(\mathbf{q}, \mathbf{v})$ represents a pairwise similarity.
\end{definition1}

A fundamental insight of our framework is that the effectiveness of both Bloom filters stems from their ability to capture set relationships through hash collision positions.  Specifically, when vector sets share similar elements, their hash functions map to overlapping positions, manifesting as either accumulated counts in the count Bloom filter or shared bit patterns in the binary Bloom filter.  This position-based collision mechanism forms the theoretical foundation for both filtering layers.  The following theorem establishes that such collision patterns in both filter types correlate with set connectivity, thereby validating the effectiveness of our dual-layer approach:

\begin{theorem}[\textbf{Collision-Similarity Relationship}]
For a query vector set $\mathbf{Q}$ and two vector sets $\mathbf{V_1}$ and $\mathbf{V_2}$, where $\mathbf{Q} \cap_h \mathbf{V}$ denotes hash collisions between elements from $\mathbf{Q}$ and $\mathbf{V}$ in either count Bloom filter or binary Bloom filter. If their collision probability satisfies:
$$
P(\mathbf{Q} \cap_h \mathbf{V}_1 \neq \emptyset)\geq
 P(\mathbf{Q} \cap_h \mathbf{V}_2 \neq \emptyset),
$$

Then:
$$
Conn(\mathbf{Q},\mathbf{V}_1) \gtrsim Conn(\mathbf{Q},\mathbf{V}_2),
$$
where \( \gtrsim \) denotes that the left-hand side is approximately greater than the right-hand side.
\end{theorem}

\begin{proof}
By the definition of collision probability:

$$
P(\mathbf{Q} \cap_h \mathbf{V}_1 \neq \emptyset) = 1 - \prod_{\mathbf{q} \in \mathbf{Q}} \prod_{\mathbf{v} \in \mathbf{V}_1} (1 - sim(\mathbf{q},\mathbf{v})),
$$

$$
P(\mathbf{Q} \cap_h \mathbf{V}_2 \neq \emptyset) = 1 - \prod_{\mathbf{q} \in \mathbf{Q}} \prod_{\mathbf{v} \in \mathbf{V}_2} (1 - sim(\mathbf{q},\mathbf{v})).
$$

The given condition implies:

$$
1 - \prod_{\mathbf{q} \in \mathbf{Q}} \prod_{\mathbf{v} \in \mathbf{V}_1} (1 - sim(\mathbf{q},\mathbf{v})) \geq
 1 - \prod_{\mathbf{q} \in \mathbf{Q}} \prod_{\mathbf{v} \in \mathbf{V}_2} (1 - sim(\mathbf{q},\mathbf{v})).
$$

Therefore:

$$
\prod_{\mathbf{q} \in \mathbf{Q}} \prod_{\mathbf{v} \in \mathbf{V}_1} (1 - sim(\mathbf{q},\mathbf{v})) \leq
 \prod_{\mathbf{q} \in \mathbf{Q}} \prod_{\mathbf{v} \in \mathbf{V}_2} (1 - sim(\mathbf{q},\mathbf{v})).
$$

Taking the logarithm of both sides:

$$
\sum_{\mathbf{q} \in \mathbf{Q}} \sum_{\mathbf{v} \in \mathbf{V}_1} \log(1 - sim(\mathbf{q},\mathbf{v})) \leq \sum_{\mathbf{q} \in \mathbf{Q}} \sum_{\mathbf{v} \in \mathbf{V}_2} \log(1 - sim(\mathbf{q},\mathbf{v})).
$$

Using the Taylor series expansion of \( \log(1-x) \) at \( x = 0 \):

$$
\log(1-x) = -x - \frac{x^2}{2} - \frac{x^3}{3} - \dots, \quad x \in [0,1]
$$

Approximating this for small \( x \):

$$
\log(1-x) \approx -x + O(x^2).
$$

Since \( sim(\mathbf{q}, \mathbf{v}) \in [0, 1] \) is a small value, we can approximate the following relationship:

$$
\sum_{\mathbf{q} \in \mathbf{Q}} \sum_{\mathbf{v} \in \mathbf{V}_1} sim(\mathbf{q}, \mathbf{v}) \gtrsim \sum_{\mathbf{q} \in \mathbf{Q}} \sum_{\mathbf{v} \in \mathbf{V}_2} sim(\mathbf{q}, \mathbf{v})
$$

This establishes:
$$
Conn(\mathbf{Q},\mathbf{V}_1) \gtrsim Conn(\mathbf{Q},\mathbf{V}_2)
$$
\end{proof}

This theoretical analysis demonstrates that higher collision probability in our filters correlates with stronger set connectivity, providing a foundation for the effectiveness of \BDLCF. The relationship between hash collisions and set connectivity validates that our dual-layer filtering mechanism effectively preserves and identifies meaningful set relationships during the search process, while the metric-independent nature of this correlation supports the framework's potential adaptability to various set distance measures.

\subsection{Discussion on Metrics Extensibility}
\BioVSS algorithm demonstrates potential adaptability to a range of set-based distance metrics beyond the Hausdorff distance, owing to the decoupled design between its filtering mechanism and the specific metric applied. Representative examples of compatible metrics include: (1) set distances defined by maximal or minimal point-pair distances \cite{toussaint1984optimal}; (2) mean aggregate distance \cite{fujita2013metrics}; and (3) various Hausdorff distance extensions \cite{ConciKubrusly2018}, among others.

This decoupling is particularly evident in the construction phase. During the creation of inverted indexes or generation of hash codes, no dependence on the specific distance metric is introduced. As a result, \BioVSS holds promise for extension to alternative distance functions. Additional discussions are provided in Section \ref{sec: Alternative_metric} and \ref{sec: Filtering Mechanism}.

\section{Experiments}
\label{sec:Experiments}

\subsection{Settings}
\begin{table}[!ht]
\scriptsize
\centering
\caption{Summary of Datasets}
\label{tab:datasetSummary}
\setlength{\tabcolsep}{2pt} % 调整列间距
\renewcommand{\arraystretch}{1} % 调整行距
\begin{small} % 调整表格字体大小
\begin{tabular}{lcccc}
\toprule
\textbf{Dataset} & \# of Vectors & \# of Vector Sets & Dim. & \# of Vectors/Set \\
\midrule
\texttt{CS} & 5,553,031 & 1,192,792 & 384 & [2, 362] \\
\texttt{Medicine} & 15,053,338 & 2,693,842 & 384 & [2, 1923] \\
\texttt{Picture} & 2,513,970 & 982,730 & 512 & [2, 9] \\
% Add more rows as needed
\bottomrule
\end{tabular}
\end{small}
% \vspace{-12pt}
\end{table}

\setcounter{footnote}{0}
\subsubsection{Datasets}
For this study, we utilize text datasets derived from the Microsoft academic graph \cite{sinha2015overview}. Two datasets were extracted from computer science and medicine domains, specifically targeting authors with a minimum of two first-author publications. Paper texts within these datasets were transformed into vector representations through the embedding model all-MiniLM-L6-v2\footnote{\url{https://www.sbert.net} \& \url{https://huggingface.co}\label{sb_url}} (widely adopted). Furthermore, we developed an image dataset employing ResNet18\footref{sb_url} for feature extraction purposes. The following presents detailed information regarding these datasets:
\begin{itemize}[leftmargin=\parindent]
\item {Computer Science Literature} (\CS): It contains 1,192,792 vector sets in the field of computer science, with 5,553,031 vectors. Each set comprises 2 to 362 vectors (Table \ref{tab:datasetSummary}).
\item {Medicine Literature} (\Medicine): It contains 2,693,842 vector sets in the field of Medicine, with 15,053,338 vectors. Each set comprises 2 to 1923 vectors (Table \ref{tab:datasetSummary}).
\item {Product Pictures} (\Picture): It contains 982,730 vector sets sourced from the AliProduct dataset \cite{DBLP_picture}, with 2,513,970 vectors. Each set represents 2 to 9 images of the same product, covering 50,000 different products (Table \ref{tab:datasetSummary}).
\end{itemize}

\subsubsection{Baselines}
We evaluate our proposed approach against multiple indexing and quantization methodologies. In particular, we utilize techniques from Facebook's Faiss library \cite{douze2024faiss} as baseline methods for comparison. The evaluated approaches include: 1) \IVFFLAT \cite{zhang2023vbase, douze2024faiss}, which utilizes an inverted file index with flat vectors for efficient processing of high-dimensional data; 2) \IndexIVFPQ \cite{ge2013optimized, jegou2010product, douze2024faiss}, which integrates inverted file indexing with product quantization to achieve vector compression and enhanced query performance; 3) \IVFScalarQuantizer \cite{douze2024faiss}, which applies scalar quantization within the inverted file index framework to balance speed and accuracy; 4) \BioVSS, our proposed approach that leverages a bio-inspired dual-layer cascade filter for efficiency enhancement through effective pruning.

Given the lack of efficient direct vector set search approaches utilizing Hausdorff distance in high-dimensional spaces, all techniques depend on centroid vectors for index construction.

\subsubsection{Evaluation Metric}
To evaluate the performance of vector set search algorithms, we utilize the recall rate across various top-$k$ values as our evaluation metric.

The recall rate at top-$k$ is formulated as:
$\text{Recall@}k = \frac{|R_k(\mathbf{Q}) \cap G_k(\mathbf{Q})|}{|G_k(\mathbf{Q})|},$
where $\mathbf{Q}$ represents a query set, $R_k(\mathbf{Q})$ denotes the algorithm's top-$k$ retrieved results for query $\mathbf{Q}$, and $G_k(\mathbf{Q})$ indicates the ground-truth (precise calculations according to Definition \ref{def:Hausdorff}). We conduct 500 queries and present the averaged recall rate.

\subsubsection{Implementation}
Our experiments were conducted on a computing platform with Intel Xeon Platinum 8352V and 512 GB of memory. The core components of our method\footnote{\url{https://github.com/whu-totemdb/biovss}\label{my_code}} were implemented in C++ and interfaced with Python.

\subsubsection{Default Parameters}
In our experiments, we focus on the following parameters: 1) the size of Bloom filter  $\{\underline{1024}, 2048\}$; 2) the number of winner-takes-all $\{16, 32, 48, \underline{64}\}$; 3) the list number of inverted index accessed $\{1, 2, \underline{3}\}$; 4) Minimum count value of inverted index $\{\underline{1}, 2\}$; 5) The size of candidate set  $\{20k, 30k, 40k, \underline{50k}\}$; and 6) The number of results returned $\{\underline{3, 5}, 10, 15, 20, 25, 30\}$. Underlined values denote default parameters in our controlled experiments.

\subsection{Storage and Construction Efficiency of Filter Structures}
\label{sec: Storage_main}
\BDLCF consists of two sparse data structures: the count Bloom filter and the binary Bloom filter. Both filters demonstrate sparse characteristics, requiring optimized storage strategies.

The storage optimization utilizes two well-known sparse formats: Coordinate (\COO) \cite{DBLP_storage1} and Compressed Sparse Row (\CSR) \cite{DBLP_storage1}. \COO format handles dynamic updates through (row, column, value) tuples. \CSR format delivers enhanced compression by preserving row pointers and column indices, which proves especially suitable for static index structures.

Table \ref{table:storage_size_comparison} presents the results on \CS. With a 1024-size Bloom filter, \CSR decreases the storage overhead from 9.1GB to 0.55GB for the count Bloom filter, reaching a 94\% reduction ratio. Similar patterns appear across different parameter settings, with \CSR uniformly surpassing \COO in storage efficiency. Results for \Medicine and \Picture display comparable features and are detailed in the Section \ref{sec: store_another}. Table \ref{table:processing_times} depicts the time overhead of different stages. While \BioHash training represents the primary computational cost at $1504s$, the construction of count and binary Bloom filters exhibits effectiveness, requiring only $24s$ and $22s$.

% \addvspace{3pt}
\begin{table}[!ht]
    \centering
    \caption{Filter Storage Comparison on \CS Dataset}
    \begin{tabular}{cccccccc}
        \toprule
        \multirow{2}{*}{\textbf{Bloom}} & \multirow{2}{*}{$L$} & \multicolumn{3}{c}{Count Bloom Space (GB)} & \multicolumn{3}{c}{Binary Bloom Space (GB)} \\
        \cline{3-8}
        & & Dense & COO & CSR & Dense & COO & CSR \\
        \hline
        \multirow{4}{*}{1024} & 16 & \multirow{4}{*}{9.1} & 1.09 & 0.55 & \multirow{4}{*}{1.14} & 0.36 & 0.19 \\
        & 32 &  & 2.02 & 1.01 &  & 0.67 & 0.34 \\
        & 48 & & 2.84 & 1.43 & & 0.95 & 0.48 \\
        & 64 & & 3.59 & 1.8 & & 1.2 & 0.6 \\
        \hline
        \multirow{4}{*}{2048} & 16 & \multirow{4}{*}{18.2} & 1.18 & 0.59 & \multirow{4}{*}{2.28} & 0.39 & 0.2 \\
        & 32 & & 2.2 & 1.1 & & 0.73 & 0.37 \\
        & 48 & & 3.13 & 1.57 & & 1.04 & 0.53 \\
        & 64 & & 4 & 2 & & 1.33 & 0.67 \\
        \bottomrule
    \end{tabular}
    \label{table:storage_size_comparison}
    
\end{table}
% \addvspace{5pt}

\begin{table}[!ht]
    \centering
    \caption{Filter Processing Time on \CS Dataset}
    \begin{tabular}{ccccc}
        \toprule
        \multirow{2}{*}{\textbf{Processing Stage}} & \BioHash & \BioHash & Count & Single \\
             & Training & Hashing & Bloom & Bloom \\
        \midrule
        Time & 1504s & 14s & 24s & 22s \\
        \bottomrule
    \end{tabular}
    \label{table:processing_times}
\end{table}

\subsection{Performance Analysis and Parameter Experiments}
In this section, we begin with preliminary experiments to evaluate the effectiveness of \NaiveBioVSS and \BioVSS, highlighting their advantages over brute-force search. We then shift focus to \BioVSS, an improved variant that integrates \BDLCF to enhance efficiency. Lastly, we conduct comprehensive parameter tuning experiments on \BioVSS to evaluate its performance and determine optimal configurations.

\begin{figure*}[!ht]
    \centering
    \includegraphics[width=0.85\textwidth]{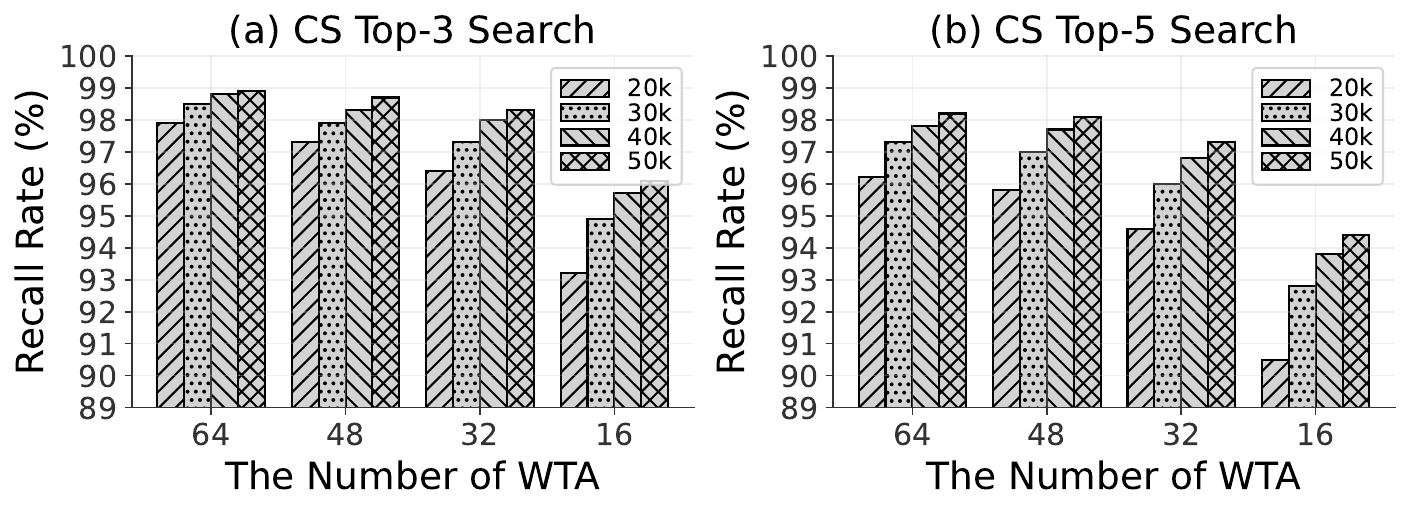}
    
    % \vspace{0.5cm}
    
    \includegraphics[width=0.85\textwidth]{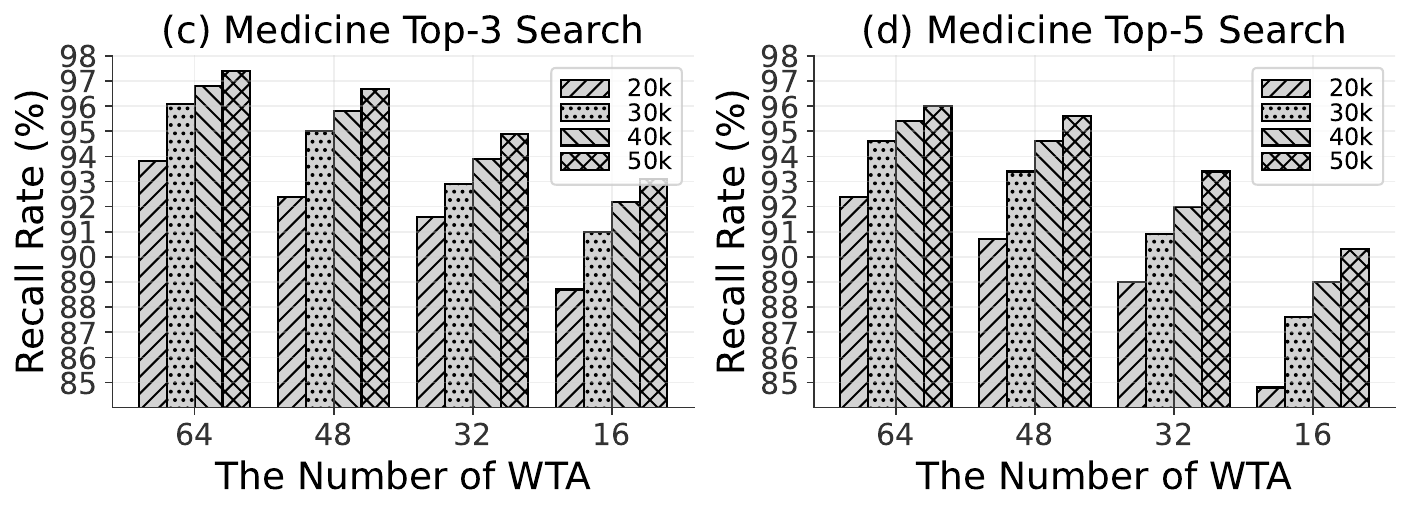}
    
    % \vspace{0.5cm}
    
    \includegraphics[width=0.85\textwidth]{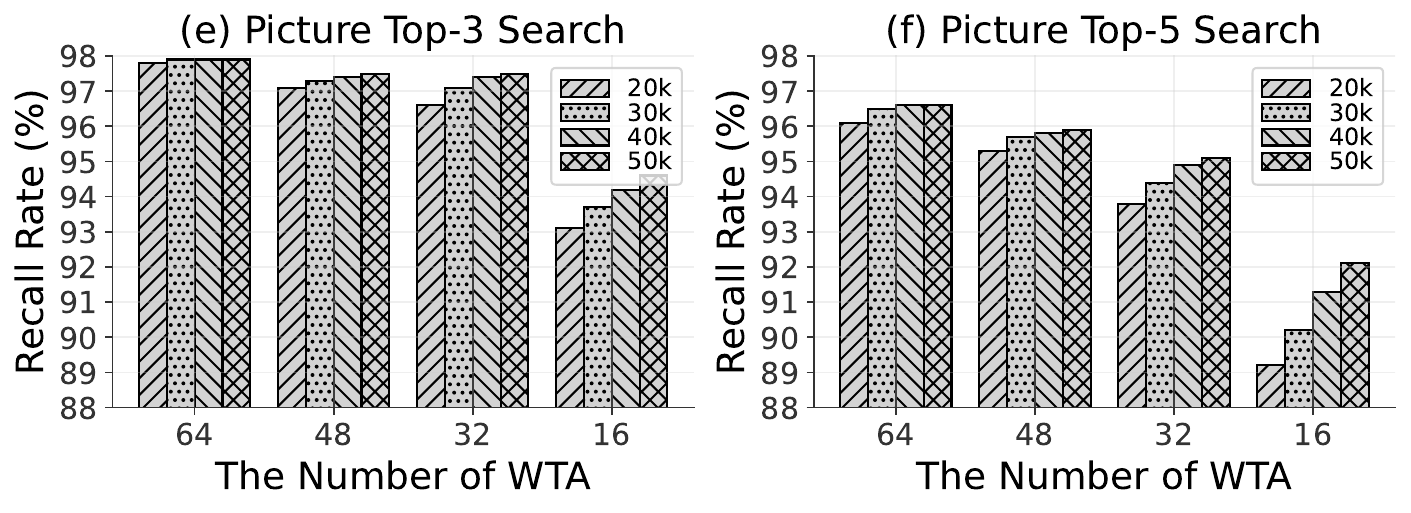}
    
    \caption{Recall Rate with Winner-Take-All Number (Bloom Size = 1024)}
    \label{fig:wta-recall-comparison-1024}
\end{figure*}

\begin{figure*}[!ht]
    \centering
    \includegraphics[width=0.85\textwidth]{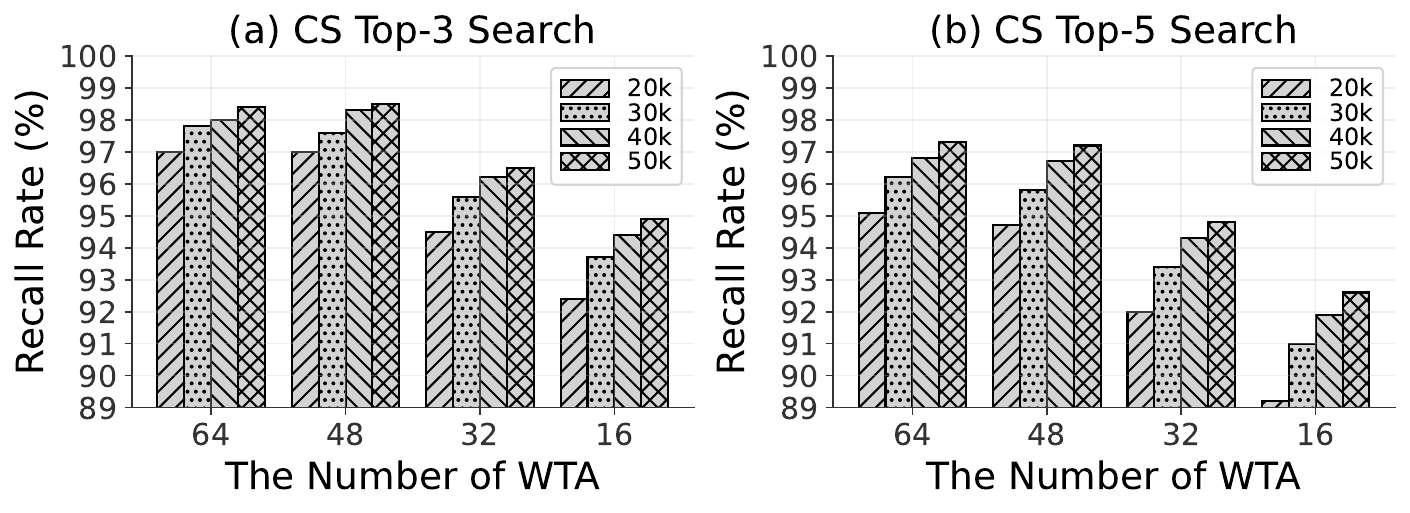}
    
    % \vspace{0.5cm}
    
    \includegraphics[width=0.85\textwidth]{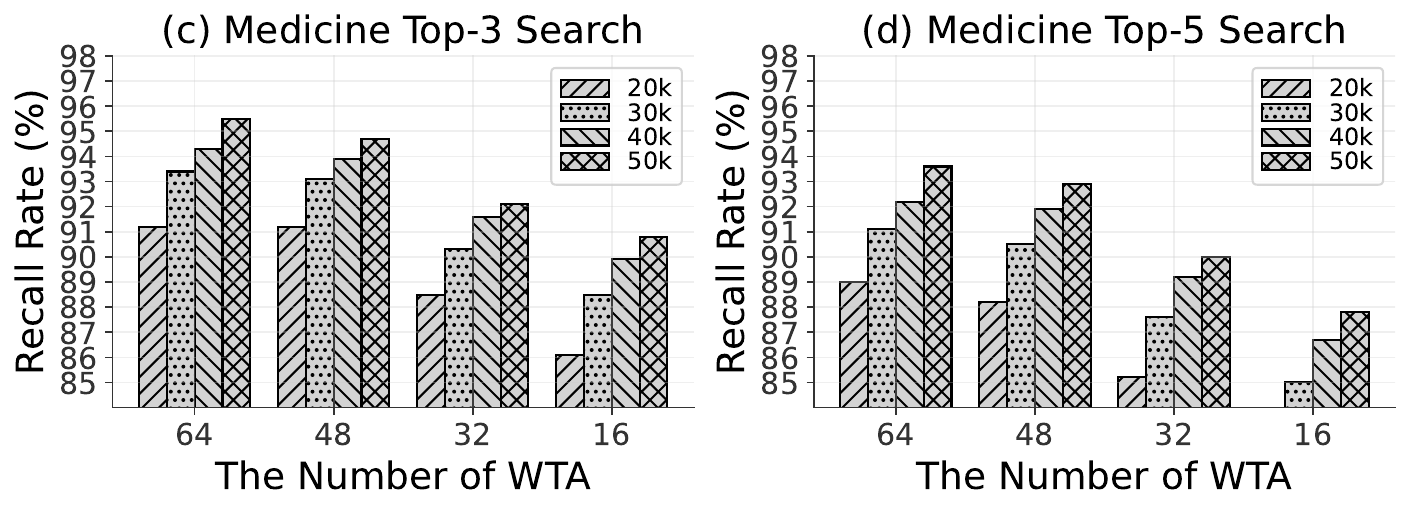}
    
    % \vspace{0.5cm}
    
    \includegraphics[width=0.85\textwidth]{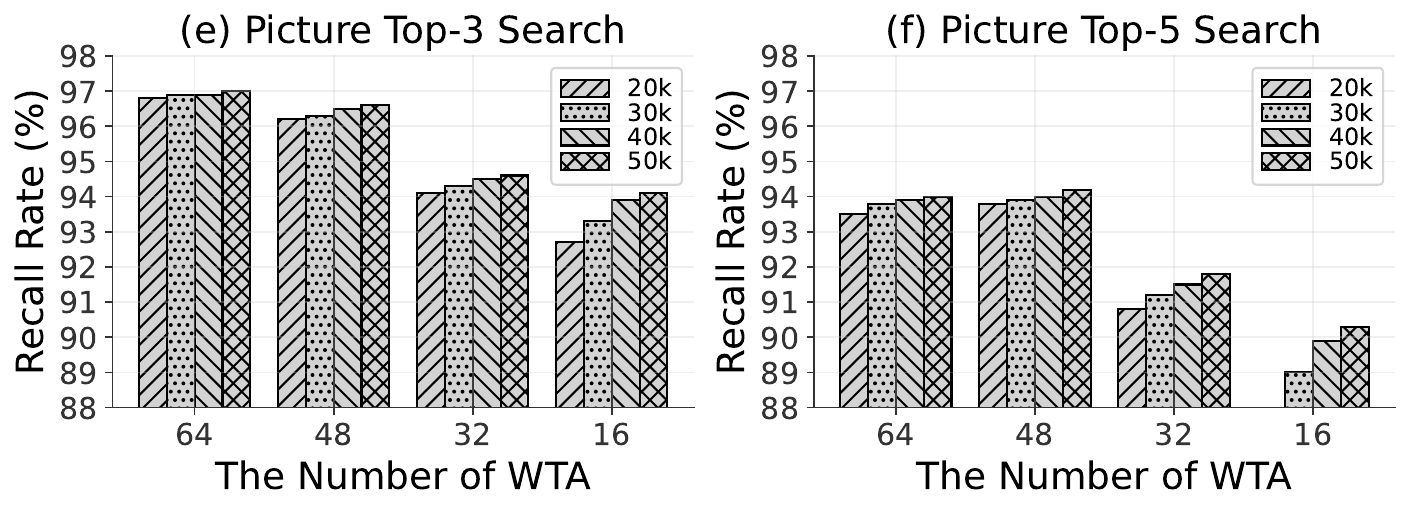}
    
    \caption{Recall Rate with Winner-Take-All Number (Bloom Size = 2048)}
    \label{fig:wta-recall-comparison-2048}
\end{figure*}
\setlength{\textfloatsep}{0cm}

\begin{figure}[t]
    \centering
    \includegraphics[width=0.85\textwidth]{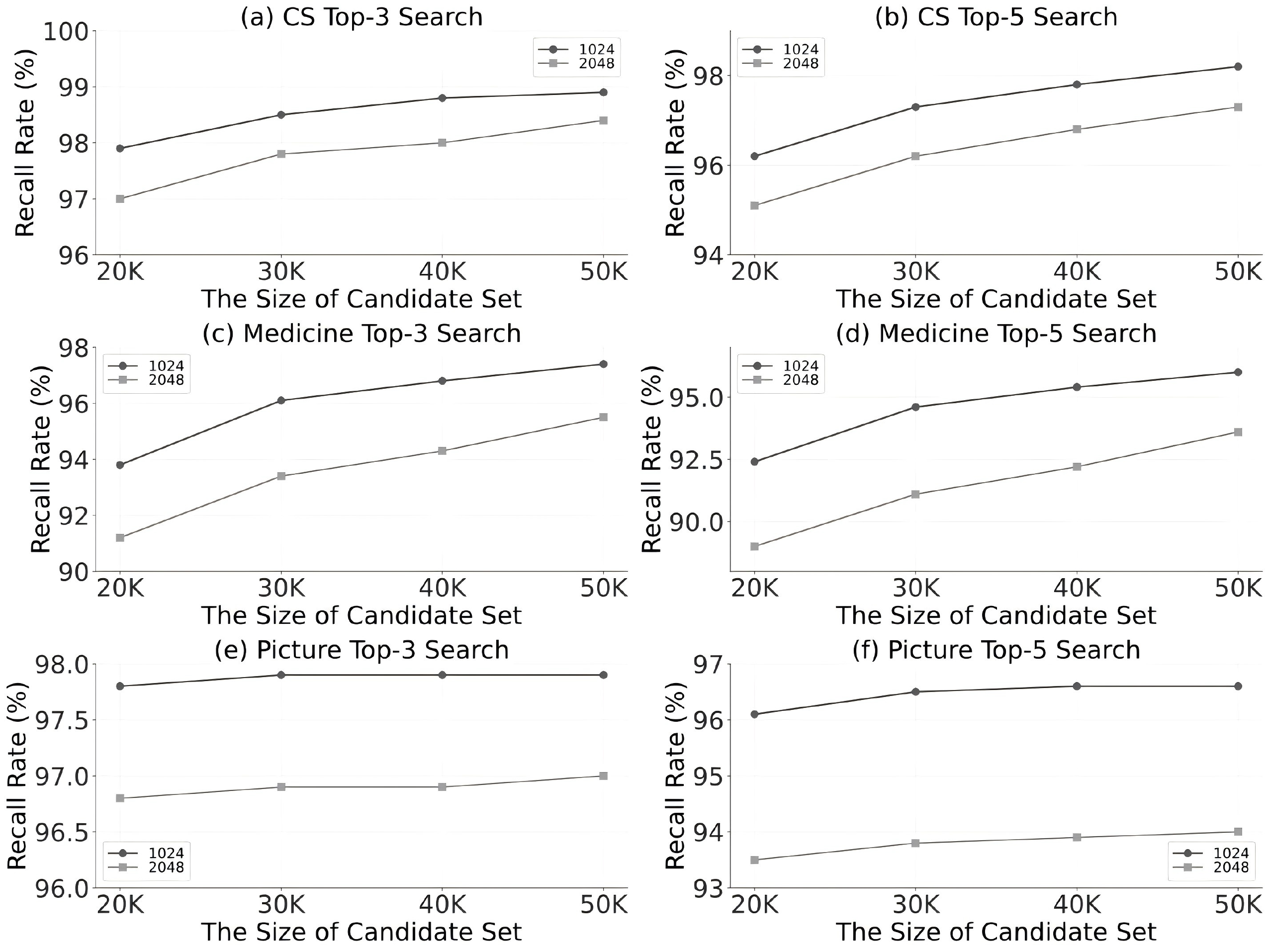}
    \captionsetup{aboveskip=2pt}  % 调整caption和图片之间的距离
    \caption{Impact of Bloom Filter Size on Recall}
    \label{fig:expand}
    \addvspace{-12pt}
\end{figure}
\setlength{\textfloatsep}{0cm}

\setlength{\textfloatsep}{0cm}
\begin{figure}[t]
    % \addvspace{-8pt}
    \centering
    \includegraphics[width=0.85\textwidth]{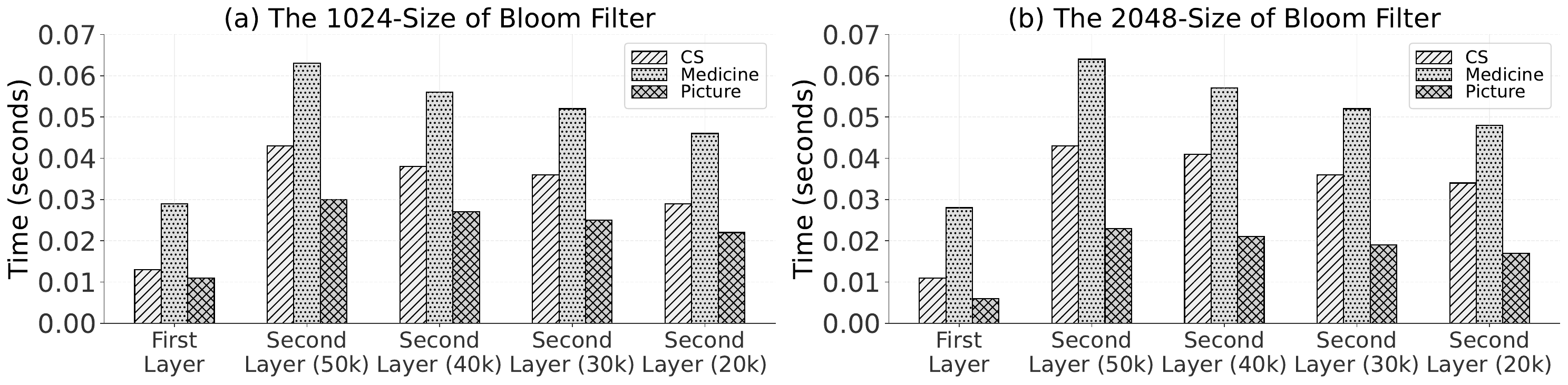}
    \captionsetup{aboveskip=2pt}  % 调整caption和图片之间的距离
    \caption{Filtering Times by the Size of Bloom Filter}
    \label{fig:filtering}
\end{figure}

\begin{table}
\centering
\small
\setlength{\tabcolsep}{3pt}
\caption{Speedup vs. Linear Scan on \CS Dataset}
\begin{tabular}{lccccc}
\toprule
\textbf{Method} & Total Time (s) & Speedup & Top-3 Recall & Top-5 Recall\\ 
\midrule
\texttt{Brute} & 9.16 & 1x & 100\% & 100\% \\
\NaiveBioVSS  & 0.73 & 12x & 97.5\% & 97.2\%\\
\BioVSS  & 0.20 & 46x & 97.9\% & 96.2\%\\
\bottomrule
\end{tabular}
\label{tab:CSSpeedup}

\vspace{5pt}

\caption{Speedup vs. Linear Scan on \Medicine Dataset}
\begin{tabular}{lccccc}
\toprule
\textbf{Method} & Total Time (s) & Speedup & Top-3 Recall & Top-5 Recall\\
\midrule
\texttt{Brute} & 16.32 & 1x & 100\% & 100\%\\
\NaiveBioVSS & 2.13 & 8x & 96.5\% & 95.8\%\\
\BioVSS & 0.24 & 78x & 93.8\% & 92.3\%\\
\bottomrule
\end{tabular}
\label{tab:MedicineSpeedup}

\vspace{5pt}

\caption{Speedup vs. Linear Scan on \Picture Dataset}
\begin{tabular}{lccccc}
\toprule
\textbf{Method} & Total Time (s) & Speedup & Top-3 Recall & Top-5 Recall\\
\midrule
\texttt{Brute} & 8.75 & 1x & 100\% & 100\%\\
\NaiveBioVSS & 0.49 & 17x & 100\% & 99.9\%\\
\BioVSS & 0.20 & 44x & 97.8\% & 96.1\%\\
\bottomrule
\end{tabular}
\label{tab:PictureSpeedup}
\end{table}

\subsubsection{Comparison with Brute-Force Search}
We begin by evaluating the performance of \NaiveBioVSS and \BioVSS against brute-force search, focusing on execution time and recall rate across the \CS, \Medicine, and \Picture datasets, each using a candidate set size of $20k$. As presented in Tables \ref{tab:CSSpeedup}, \ref{tab:MedicineSpeedup}, and \ref{tab:PictureSpeedup}, both \NaiveBioVSS and \BioVSS achieve substantial speed improvements over brute-force search while preserving high recall levels. On the \CS dataset, \BioVSS delivers a notable 46× speedup, reducing the execution time from 9.16 seconds to 0.20 seconds, with top-$3$ and top-$5$ recall rates reaching 97.9\% and 96.2\%, respectively. Even greater acceleration is observed on the \Medicine dataset, where \BioVSS achieves a 78× speedup, lowering execution time from 16.32 seconds to just 0.24 seconds. On the \Picture dataset, \BioVSS reduces execution time from 8.75 seconds to 0.20 seconds—yielding a 44× improvement—while maintaining 97.8\% and 96.1\% recall for top-3 and top-5 results, respectively. These significant gains primarily stem from the filtering mechanism incorporated in \BioVSS, which effectively mitigates the global scanning overhead inherent in \NaiveBioVSS.

\subsubsection{Parameter Study of \BioVSS}
\label{sec: Parameter_exp}
To optimize the performance of \BioVSS, we conducted a comprehensive parameter study. We examined several key parameters: the number of winner-takes-all, the size of the Bloom filter, the list number of inverted index accessed, and the minimum count value of the inverted index.

\myparagraph{Number of Winners-Take-All}
The number of winner-takes-all (see Definition \ref{def: FlyHash}) in \BioVSS corresponds to the number of hash functions. As shown in Figures \ref{fig:wta-recall-comparison-1024} and \ref{fig:wta-recall-comparison-2048}, increasing this parameter from 16 to 48 led to significant performance improvements across \CS, \Medicine, and \Picture datasets for Bloom filter sizes of 1024 and 2048. On \CS dataset, with a 1024-size Bloom filter and 20k candidates, recall improved by 5.7\% when increasing the parameter from 16 to 48. However, the performance gain plateaued between 48 and 64, with only a 0.4\% improvement in the same scenario. Similarly, \Medicine and \Picture datasets exhibit the same trend when increasing the parameter. This trend was consistent across different datasets and Bloom filter sizes. This improvement can be attributed to the enhanced distance-preserving property of the hash code as the number of winner-takes-all increases. Hence, the default number of winner-takes-all is set to 64.

\myparagraph{Size of the Bloom Filter}
The size of the Bloom filter directly influences the length of set sketches and the list number in the inverted index for \BioVSS. As illustrated in Figure \ref{fig:expand}, our experiments explore filter sizes of 1024 and 2048, revealing that a Bloom filter with a size of 1024 achieves optimal recall rates across all candidate numbers. On \CS dataset, the 1024-size filter configuration yields a $98.9\%$ recall rate with a candidate set of $50k$, while maintaining a robust $98\%$ recall even with a reduced set of $20k$ candidates. Similarly, for \Medicine and \Picture datasets, the 1024-size filter achieves optimal recall rates. The effectiveness of this filter size can be attributed to its capacity to capture discriminative features without over-emphasizing local characteristics, thereby providing a good trade-off between specificity and generalization in this process. As demonstrated in Figure \ref{fig:filtering}, The various sizes of the Bloom filters exhibit low latency below 70 milliseconds. Hence, the default size of the Bloom filter is set to 1024.

\myparagraph{List Number of Inverted Index Accessed}
The list number of inverted index accessed determines the search range in BioVSS, with larger values leading to broader searches. As shown in Table \ref{tab:access_number}, increasing this parameter from 1 to 3 significantly improves recall rates while modestly increasing processing time. For \CS dataset with a 1024-size Bloom filter, the top-3 recall improves from 92.9\% to 98.9\%, and the top-5 recall from 91.1\% to 98.2\%, with processing time increasing from 0.008s to 0.013s. Similar improvements are observed with a 2048-size Bloom filter. \Medicine and \Picture datasets show the same trends, with slightly lower recall rates but similar time increases. Notably, the recall rate gain diminishes when moving from 2 to 3 accesses. Hence, the default list number of inverted index accessed is 3.

\begin{table}[!h]
    \centering
    \footnotesize
    \setlength{\tabcolsep}{3pt}
    \captionsetup{skip=2pt}  % 调整caption和表格之间的间距
    
    \caption{List Access Number for Top-3 (T-3) and Top-5 (T-5)}
    \begin{tabular}{@{}l*{9}{c}@{}}
        \toprule
        \multirow{3}{*}{\textbf{Dataset}} & \multicolumn{3}{c}{1 Access} & \multicolumn{3}{c}{2 Access} & \multicolumn{3}{c}{3 Access} \\
        \cmidrule(lr){2-4} \cmidrule(lr){5-7} \cmidrule(lr){8-10}
        & \multicolumn{2}{c}{Recall (\%)} & Time & \multicolumn{2}{c}{Recall (\%)} & Time & \multicolumn{2}{c}{Recall (\%)} & Time \\
        & T-3 & T-5 & (s) & T-3 & T-5 & (s) & T-3 & T-5 & (s) \\
        \midrule
        \CS (1024) & 92.9 & 91.1 & 0.008 & 98.0 & 97.1 & 0.012 & 98.9 & 98.2 & 0.013 \\
        \CS (2048) & 91.9 & 89.4 & 0.008 & 97.6 & 96.4 & 0.011 & 98.4 & 97.3 & 0.011 \\
        \Medicine (1024) & 92.8 & 90.8 & 0.014 & 96.9 & 95.6 & 0.018 & 97.4 & 96.0 & 0.029 \\
        \Medicine (2048) & 90.3 & 86.8 & 0.014 & 94.3 & 92.2 & 0.019 & 95.5 & 93.6 & 0.028 \\
        \Picture (1024) & 85.1 & 76.3 & 0.017 & 94.8 & 91.2 & 0.024 & 97.9 & 96.6 & 0.030 \\
        \Picture (2048) & 81.5 & 73.6 & 0.013 & 93.1 & 88.0 & 0.018 & 97.0 & 94.0 & 0.023 \\
        \bottomrule
    \end{tabular}
    \label{tab:access_number}
\end{table}

\myparagraph{Minimum Count Value of Inverted Index}
The minimum count value of the inverted index determines the threshold for accessing items in the inverted index lists. As illustrated in Figure \ref{fig:count}, this parameter impacts the trade-off between recall. Setting the value to 0 requires traversing all items. A value of 1 leverages the hash codes' sparsity, eliminating most irrelevant items while maintaining high recall. Increasing to 2 further reduces candidates but at the cost of the recall rate. For \CS dataset with a 1024-size Bloom filter, the top-3 recall decreases from 98.9\% at count 1 to 96.5\% at count 2. Top-5 recall shows a similar trend, declining from 98.2\% to 95.7\%. \Medicine and \Picture datasets exhibit similar trends, with comparable drops in top-3 and top-5 recall rates. Thus,
we set the default minimum count to 1.

\begin{table}
    \centering
    \footnotesize
    \setlength{\tabcolsep}{3.5pt}  % 调小了列间距以适应新增列
    \captionsetup{skip=2pt}
    \caption{Minimum Count Value with Recall Rate and Time}
\begin{tabular}{@{}lccccccc@{}}
\toprule
\multirow{2}{*}{\textbf{Dataset}} & \multicolumn{3}{c}{Minimum Count $=$ 1} & \multicolumn{3}{c}{Minimum Count $=$ 2} \\
\cmidrule(lr){2-4} \cmidrule(lr){5-7}
 & Top-3 & Top-5 & Total Time & Top-3 & Top-5 & Total Time \\
\midrule
\CS (1024) & 98.9\% & 98.2\% & 0.42s & 96.5\% & 95.1\% & 0.45s \\
\CS (2048) & 98.4\% & 97.3\% & 0.45s & 95.1\% & 94.0\% & 0.44s \\
\Medicine (1024) & 97.4\% & 96.0\% & 0.51s & 93.3\% & 91.8\% & 0.50s \\
\Medicine (2048) & 95.5\% & 93.6\% & 0.48s & 89.7\% & 87.0\% & 0.47s \\
\Picture (1024) & 97.9\% & 96.6\% & 0.42s & 92.0\% & 87.3\% & 0.45s \\
\Picture (2048) & 97.0\% & 94.0\% & 0.43s & 87.9\% & 81.6\% & 0.43s \\
\bottomrule
\end{tabular}
\label{fig:count}
\end{table}

\myparagraph{Impact of Embedding Models on Algorithm Performance}
To assess the robustness of \NaiveBioVSS across various vector representations, we conducted experiments using different embedding models. These models inherently produce vectors of varying dimensions, allowing us to evaluate the method's performance across different vector lengths. The experiments were designed to evaluate two critical aspects of embedding model impact. First, the performance consistency was tested using different embedding models within the same modality but with varying vector dimensions. Second, the method's versatility was examined across different modalities while maintaining consistent vector dimensions. \CS and \Picture datasets represent text and image modalities respectively, enabling these comparative analyses. To validate the impact of different vector dimensions within the same modality, two text embedding models were applied to \CS dataset. All-MiniLM-L6-v2\footref{sb_url} and distiluse-base-multilingual-cased-v2\footref{sb_url} from Hugging Face generate 384-dimensional and 512-dimensional vectors. Table \ref{tab:embedding_model} elucidates the efficacy and robustness of \NaiveBioVSS across diverse embedding dimensionalities. For \CS dataset with a Bloom filter cardinality of 1024, both the 384-dimensional MiniLM and 512-dimensional DistilUse models exhibit exceptional recall performance. MiniLM achieves Top-3 and Top-5 recall rates of 98.9\% and 98.2\%, while DistilUse demonstrates 92.3\% and 90.8\%. Significantly, the computational latencies are nearly equivalent irrespective of vector dimensionality, corroborating the dimension-agnostic nature of the search algorithm. The algorithm maintains similarly high recall rates when increasing the Bloom filter size to 2048. Notably, the search process exhibits dimension-invariant computational complexity, ensuring efficient performance across varying dimensional scales. To validate the impact of different modalities with the same vector dimension, embedding models generating 512-dimensional vectors were applied to \CS and \Picture datasets. Specifically, with a Bloom filter cardinality of 1024, both \CS and \Picture datasets demonstrate high recall rates (\>90\% for Top-5), with comparable computational latencies (0.4s+). This performance is maintained when scaling to a size of 2048, evidencing the algorithm's robustness across data modalities.

\begin{table}[h]
\centering
\footnotesize
\setlength{\tabcolsep}{3pt}
\captionsetup{skip=2pt}
\caption{Impact of Different Embedding Methods}
\begin{tabular}{@{}l*{5}{c}@{}}
\toprule
\multirow{2}{*}{\textbf{Dataset}} & \textbf{Embedding} & \textbf{Embedding} & \multicolumn{2}{c}{\textbf{Recall (\%)}} & \textbf{Time} \\
 & \textbf{Model} & \textbf{Dimension} & \textbf{Top-3} & \textbf{Top-5} & \textbf{(s)} \\
\midrule
\CS (1024) & MiniLM & 384 & 98.9 & 98.2 & 0.42\\
\CS (1024) & DistilUse & 512 & 92.3& 90.8& 0.46\\
\Picture (1024) & ResNet18 & 512 & 97.9 & 96.6 & 0.42 \\
\midrule
\CS (2048) & MiniLM & 384 & 98.4 & 97.3 & 0.43\\%
\CS (2048) & DistilUse & 512 & 85.2& 81.0& 0.45\\
\Picture (2048) & ResNet18 & 512 & 97.0 & 94.0 & 0.43\\%
\bottomrule
\end{tabular}
\label{tab:embedding_model}
\end{table}

\myparagraph{Impact of Top-k on Result Quality}
To evaluate the impact of the top-$k$ parameter on result quality, we conducted experiments using default parameters for both \NaiveBioVSS and \BioVSS on \CS, \Medicine, and \Picture datasets. Table \ref{tab:combined_topk} presents the recall rates for various top-$k$ values ranging from 3 to 30. On \CS dataset, \NaiveBioVSS demonstrates consistent performance, maintaining a high recall rate above 98\% across all top-$k$ values. \BioVSS shows slightly higher recall for top-3 (98.9\%) but experiences a gradual decrease as $k$ increases, reaching 96.4\% for top-30. \Medicine dataset reveals a similar trend, with \NaiveBioVSS maintaining high recall rates (99.1\% for top-3, decreasing slightly to 96.9\% for top-30), while \BioVSS shows a decline (from 97.5\% for top-3 to 92.3\% for top-30). \Picture dataset exhibits a similar trend. The slight performance degradation observed in \BioVSS for larger $k$ values suggests a trade-off between efficiency and recall, which may be attributed to its more aggressive filtering mechanism. Notably, both methods maintain high recall rates (above 92\%) even for larger $k$ values, demonstrating their effectiveness in retrieving relevant results across various retrieval scenarios. The filtering mechanism of \BioVSS has brought significant improvements in query efficiency. Smaller top-$k$ values are typically more valuable and often yield the most relevant results in many applications. Consequently, we will focus on \BioVSS in our subsequent experiments. 

\begin{table}[h]
\vspace{4pt}
\centering
\footnotesize
\setlength{\tabcolsep}{3pt}
\caption{Recall Rate for Top-$k$ Across Different Datasets}
\label{tab:combined_topk}
\begin{tabular}{lccccccc}
\toprule
\textbf{Method} & Top-3 & Top-5 & Top-10 & Top-15 & Top-20 & Top-25 & Top-30 \\
\midrule
\multicolumn{8}{c}{\CS Dataset} \\
\midrule
\texttt{BioVSS} & 98.6\% & 98.7\% & 98.5\% & 98.5\% & 98.4\% & 98.2\% & 98.1\% \\
\texttt{BioVSS++} & 98.9\% & 98.2\% & 97.5\% & 97.2\% & 96.9\% & 96.7\% & 96.4\% \\
\midrule
\multicolumn{8}{c}{\Medicine Dataset} \\
\midrule
\texttt{BioVSS} & 99.1\% & 98.8\% & 97.8\% & 97.5\% & 97.2\% & 97.0\% & 96.9\% \\
\texttt{BioVSS++} & 97.4\% & 96.0\% & 94.8\% & 93.9\% & 93.2\% & 92.7\% & 92.3\% \\
\midrule
\multicolumn{8}{c}{\Picture Dataset} \\
\midrule
\texttt{BioVSS} & 100\% & 99.9\% & 99.8\% & 99.7\% & 99.6\% & 99.6\% & 99.5\% \\
\texttt{BioVSS++} & 99.9\% & 99.9\% & 90.0\% & 80.0\% & 85.0\% & 80.0\% & 73.3\% \\
\bottomrule
\end{tabular}
\vspace{-4pt}
\end{table}

\myparagraph{Query Time Analysis}
Experiments evaluated query efficiency across candidate sets ranging from 20k to 50k. As illustrated in Table~\ref{table: all_time}, query time scales approximately linearly with candidate count. With the configuration of WTA=64 and Bloom filter size=1024, the method achieves consistent and efficient query times (0.44s-0.51s for 50k candidates) across all datasets. While reducing WTA hash count to 16 decreases query time by up to 15\%, such reduction compromises recall performance as shown in Figures~\ref{fig:wta-recall-comparison-1024} and \ref{fig:wta-recall-comparison-2048}. Additionally, doubling the Bloom filter size to 2048 offers minimal efficiency gains ($\leq 0.03s$ improvement), making the added memory overhead unjustifiable. These observations support selecting 64 and 1024 as the optimal configurations balancing efficiency and effectiveness.

\begin{table*}[t]
\centering
\caption{Query Time (s) with Different Bloom Filter and WTA for Top-3\&5}
\label{table: all_time}
\resizebox{\textwidth}{!}{% 自动缩放至文本宽度
\footnotesize
\begin{tabular}{l*{16}{c}} % 使用紧凑列格式
\toprule
\multirow{2}{*}{Dataset} & \multicolumn{4}{c}{B=1024, W=64} & \multicolumn{4}{c}{B=1024, W=48} & \multicolumn{4}{c}{B=1024, W=32} & \multicolumn{4}{c}{B=1024, W=16} \\
\cmidrule(lr){2-5} \cmidrule(lr){6-9} \cmidrule(lr){10-13} \cmidrule(lr){14-17}
& 50k & 40k & 30k & 20k & 50k & 40k & 30k & 20k & 50k & 40k & 30k & 20k & 50k & 40k & 30k & 20k \\
\midrule
\CS & 0.44 & 0.35 & 0.29 & 0.21 & 0.46 & 0.48 & 0.38 & 0.29 & 0.44 & 0.46 & 0.31 & 0.28 & 0.43 & 0.45 & 0.37 & 0.27 \\
\Medicine & 0.51 & 0.40 & 0.34 & 0.24 & 0.49 & 0.41 & 0.32 & 0.25 & 0.47 & 0.39 & 0.28 & 0.22 & 0.45 & 0.36 & 0.28 & 0.19 \\
\Picture & 0.44 & 0.36 & 0.28 & 0.20 & 0.45 & 0.36 & 0.29 & 0.20 & 0.45 & 0.39 & 0.30 & 0.21 & 0.45 & 0.38 & 0.29 & 0.21 \\
\midrule

\multirow{2}{*}{Dataset} & \multicolumn{4}{c}{B=2048, W=64} & \multicolumn{4}{c}{B=2048, W=48} & \multicolumn{4}{c}{B=2048, W=32} & \multicolumn{4}{c}{B=2048, W=16} \\
\cmidrule(lr){2-5} \cmidrule(lr){6-9} \cmidrule(lr){10-13} \cmidrule(lr){14-17}
& 50k & 40k & 30k & 20k & 50k & 40k & 30k & 20k & 50k & 40k & 30k & 20k & 50k & 40k & 30k & 20k \\
\midrule
\CS & 0.45 & 0.34 & 0.28 & 0.20 & 0.44 & 0.33 & 0.27 & 0.20 & 0.43 & 0.33 & 0.27 & 0.20 & 0.39 & 0.32 & 0.26 & 0.19 \\
\Medicine & 0.48 & 0.40 & 0.31 & 0.22 & 0.47 & 0.39 & 0.29 & 0.22 & 0.46 & 0.38 & 0.29 & 0.21 & 0.43 & 0.37 & 0.28 & 0.19 \\
\Picture & 0.43 & 0.35 & 0.27 & 0.20 & 0.46 & 0.37 & 0.29 & 0.20 & 0.43 & 0.37 & 0.28 & 0.20 & 0.43 & 0.37 & 0.28 & 0.20 \\
\bottomrule
\end{tabular}}
\begin{minipage}{\textwidth}
\footnotesize
\textsuperscript{1}B = Bloom Filter Size, W = WTA hash count\\
\end{minipage}
\end{table*}

\subsection{Comparison Experiment}
To evaluate the performance of the \BioVSS method, we conducted comparative evaluations against baseline algorithms on the \CS, \Medicine, and \Picture datasets. The experiments focus on analyzing query time and recall rate. The following section details the results and corresponding analysis.

Figures \ref{fig:recall-comparison}(a) and (b) illustrate the experimental outcomes on the \CS dataset. From the comparison of time-recall trends, it is clear that the proposed \BioVSS algorithm consistently surpasses the other three baselines. In order to manage query time, we vary the candidate set size. Results indicate that recall rates improve across all methods as query time increases. Notably, \BioVSS achieves a recall of 98.9\% within only 0.2 seconds for the top-3 query, and maintains a recall of 98.2\% at 0.47 seconds. Even under the top-5 setting, its recall remains above 90\%. It is also observed that both \IVFScalarQuantizer and \IVFFLAT exhibit sharp gains in recall as query time grows, whereas the increase for \IndexIVFPQ is more modest. This may be attributed to inherent limitations in the product quantization encoding technique. Specifically, product quantization segments high-dimensional vectors into multiple subvectors, each of which is quantized independently. Although this enhances compression, it can also introduce information loss, restricting further recall improvements.

Figures \ref{fig:recall-comparison}(c) and (d) report the query performance on the \Medicine dataset. Given the larger data volume, the overall query time increases accordingly. Among the baselines, \IVFScalarQuantizer shows a steeper recall growth curve as time progresses, suggesting greater sensitivity to candidate set size in large-scale settings. In contrast, our \BioVSS method continues to deliver strong performance, even when candidate sets are relatively small. Figures \ref{fig:recall-comparison}(e) and (f) present the results on the \Picture dataset, where the performance trends align closely with those observed in the \CS dataset.

\begin{figure*}[h]
    \centering
    \includegraphics[width=\textwidth]{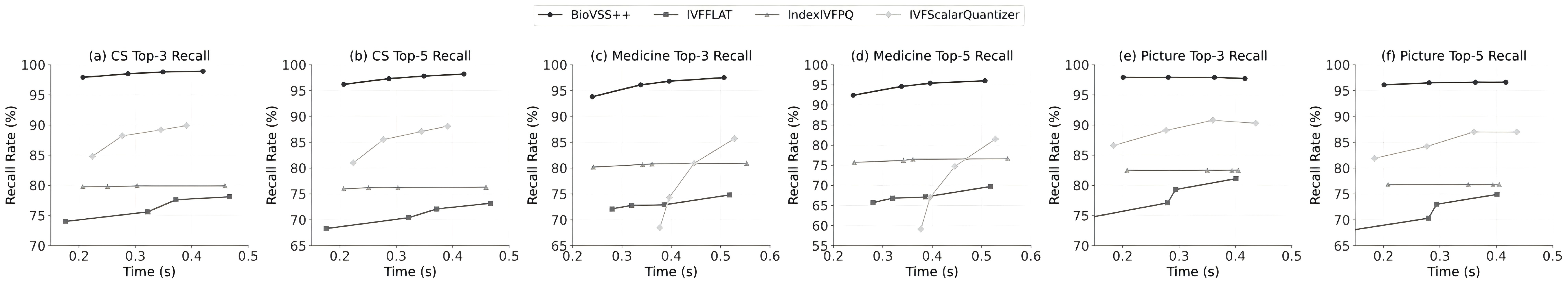}
    \captionsetup{aboveskip=2pt}  % 调整caption和图片之间的距离
    \caption{Recall Rate Comparison for Different Methods}
    \label{fig:recall-comparison}
\end{figure*}

\subsection{Supplementary Experiments}
This section presents additional experiments and analyses further to validate the performance and versatility of \NaiveBioVSS. We explore three key aspects: storage analysis on \Medicine and \Picture datasets, exploration of alternative distance metrics, and performance analysis via \BioHash iteration.

\subsubsection{Storage Analysis on Medicine and Picture Datasets}
\label{sec: store_another}
Tables \ref{table:medicine_storage_size_comparison} and \ref{table:picture_storage_size_comparison} present the storage efficiency analysis on \Medicine and \Picture datasets. These results corroborate the findings from \CS, demonstrating consistent storage reduction patterns across different data domains.

\addvspace{5pt}
\begin{table}[h]
    \centering
    \caption{Filter Storage Comparison on \Medicine Dataset}
    \begin{tabular}{cccccccc}
        \toprule
        \multirow{2}{*}{\textbf{Bloom}} & \multirow{2}{*}{$L$} & \multicolumn{3}{c}{Count Bloom Space (GB)} & \multicolumn{3}{c}{Binary Bloom Space (GB)} \\
        \cline{3-8}
        & & Dense & COO & CSR & Dense & COO & CSR \\
        \hline
        \multirow{4}{*}{1024} & 16 & \multirow{4}{*}{7.5} & 0.57 & 0.29 & \multirow{4}{*}{0.94} & 0.19 & 0.1 \\
        & 32 &  & 1.13 & 0.57 &  & 0.38 & 0.19 \\
        & 48 &  & 1.67 & 0.84 &  & 0.56 & 0.28 \\
        & 64 &  & 2.18 & 1.1 &  & 0.73 & 0.37 \\
        \hline
        \multirow{4}{*}{2048} & 16 & \multirow{4}{*}{15} & 0.61 & 0.31 & \multirow{4}{*}{1.87} & 0.2 & 0.11 \\
        & 32 &  & 1.18 & 0.6 &  & 0.39 & 0.2 \\
        & 48 &  & 1.73 & 0.87 &  & 0.58 & 0.29 \\
        & 64 &  & 2.29 & 1.15 &  & 0.76 & 0.38 \\
        \bottomrule
    \end{tabular}
    \label{table:medicine_storage_size_comparison}
\end{table}
\addvspace{10pt}

\begin{table}[h]
    \centering
    \caption{Filter Storage Comparison on \Picture Dataset}
    \begin{tabular}{cccccccc}
        \toprule
        \multirow{2}{*}{\textbf{Bloom}} & \multirow{2}{*}{$L$} & \multicolumn{3}{c}{Count Bloom Space (GB)} & \multicolumn{3}{c}{Binary Bloom Space (GB)} \\
        \cline{3-8}
        & & Dense & COO & CSR & Dense & COO & CSR \\
        \hline
        \multirow{4}{*}{1024} & 16 & \multirow{4}{*}{20.55} & 2.79 & 1.41 & \multirow{4}{*}{2.57} & 0.93 & 0.48 \\
        & 32 &  & 5.14 & 2.58 &  & 1.71 & 0.87 \\
        & 48 &  & 7.26 & 3.64 &  & 2.42 & 1.22 \\
        & 64 &  & 9.22 & 4.62 &  & 3.07 & 1.55 \\
        \hline
        \multirow{4}{*}{2048} & 16 & \multirow{4}{*}{41.1} & 3.01 & 1.51 & \multirow{4}{*}{5.14} & 1 & 0.51 \\
        & 32 &  & 5.6 & 2.81 &  & 1.87 & 0.94 \\
        & 48 &  & 8.02 & 4.02 &  & 2.67 & 1.35 \\
        & 64 &  & 10.71 & 5.37 &  & 3.57 & 1.8 \\
        \bottomrule
    \end{tabular}
    \label{table:picture_storage_size_comparison}
\end{table}
\addvspace{5pt}

\subsubsection{Exploration of Alternative Distance Metrics}
\label{sec: Alternative_metric}
while the main text demonstrates the effectiveness of \BioVSS using Hausdorff distance, the framework’s applicability extends to other set-based distance metrics. To further validate this extensibility, additional experiments were conducted using alternative distance measures.

Our experiments compared \BioVSS against \DESSERT\footnote{https://github.com/ThirdAIResearch/Dessert. Without IVF indexing.} under various parameter configurations (Table \ref{tab:dessert_comparison}). Using the MeanMin distance metric, \BioVSS achieved a Top-3 recall of 59.0\% with a query time of 0.46 seconds. These results demonstrate reasonable efficiency beyond the Hausdorff setting.
The performance disparity between Hausdorff and MeanMin distances arises from their aggregation mechanisms. Hausdorff utilizes a three-level structure ($min$-\>$max$-\>$max$), whereas MeanMin employs a two-level aggregation ($min$-\>$mean$). Such three-level aggregation creates more distinctive distance distributions. Consequently, similar sets exhibit higher collision probabilities, while dissimilar sets are more effectively separated.

The extensibility to different metrics stems from \BioVSS's decoupled design. The filter structure operates independently of the specific distance metric. This architectural choice enables metric flexibility while maintaining the core filtering mechanisms.

\begin{table}[t]
  \centering
  \caption{Performance Comparison of MeanMin}
  \begin{tabular}{cccc}
    \toprule
    Method & Top-3 & Top-5 & Time \\
    \midrule
    \DESSERT (tables=32, hashes\_per\_t=6) & 45.9\% & 35.8\% & 0.21s \\
    \DESSERT (tables=32, hashes\_per\_t=12) & 40.3\% & 28.8\% & 0.56s \\
    \DESSERT (tables=24, hashes\_per\_t=6) & 42.3\% & 32.6\% & 0.18s \\
    \DESSERT (num\_t=24, hashes\_per\_t=12) & 38.7\% & 27.6\% & 0.36s \\
    \BioVSS (default\_parameter) & 59.0\% & 51.6\% & 0.46s \\
    \bottomrule
  \end{tabular}
  \label{tab:dessert_comparison}
\end{table}

\subsubsection{Performance Analysis via BioHash Iteration}
\label{sec: bio_iter}
\BioHash algorithm serves as a crucial component in our framework, with its reliability directly impacting system performance. While \BioHash has been validated in previous research \cite{FlyLSH_ryali2020bio}, a rigorous examination of its behavior within our specific application domain is necessary.

The iteration count parameter plays a vital role in \BioHash's computational process, specifically in determining learning efficacy. At its core, \BioHash implements a normalized gradient descent optimization approach. The magnitude of parameter updates functions as a key metric for quantifying learning dynamics. This magnitude is defined as:
$$
M_t = \max_{i,j} |{\Delta W_{ij}^t}|,
$$
where $M_t$ represents the update magnitude at iteration $t$, and $\Delta W_{ij}^t$ denotes the weight change for the synaptic connection between neurons $i$ and $j$. This metric effectively captures the most substantial parametric modifications occurring within the network during each training iteration ($batch\_size=10k$).

\begin{figure}[t]
    \centering
    % 第一行的两张图片
    \subfloat{
        \includegraphics[width=0.48\textwidth]{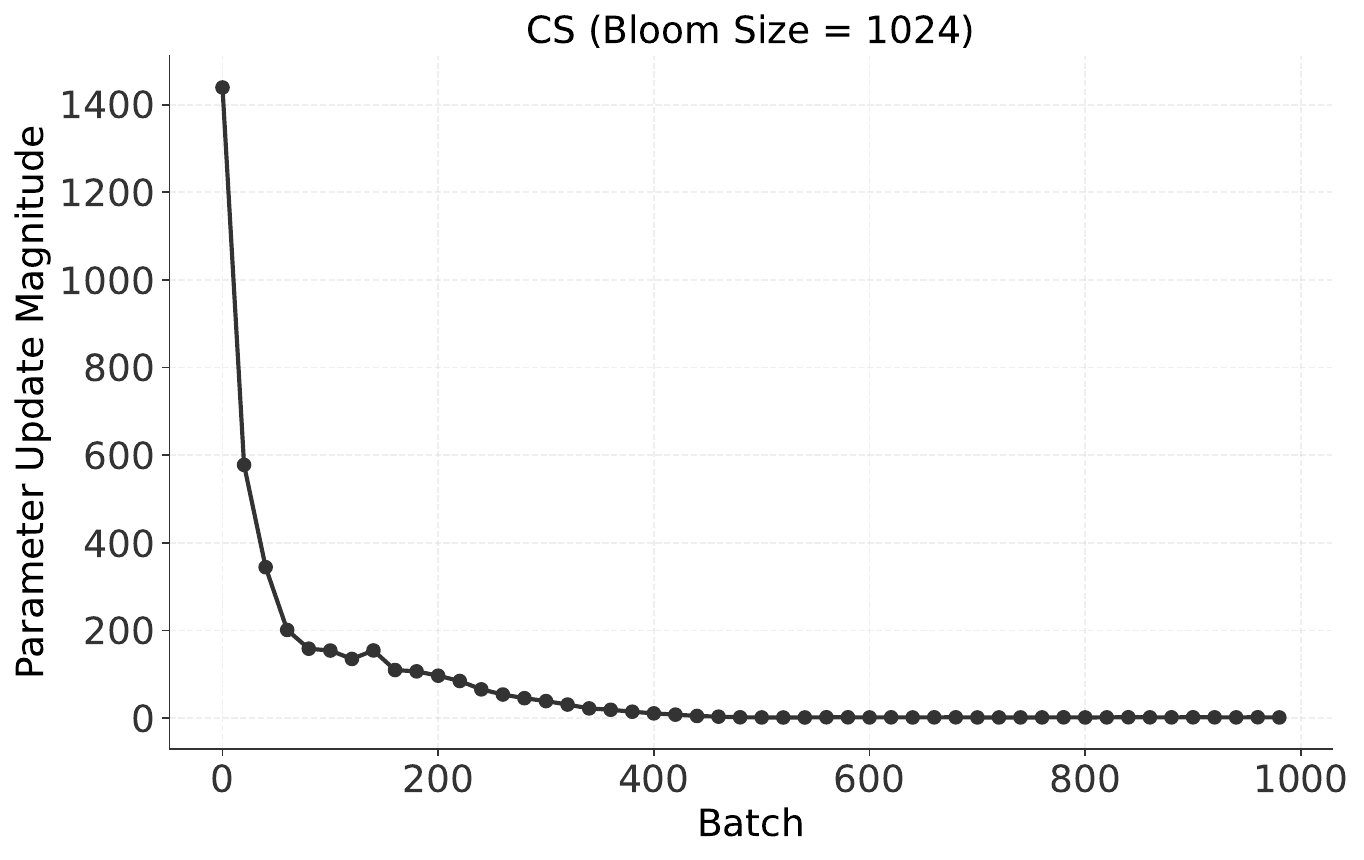}
    }
    \subfloat{
        \includegraphics[width=0.48\textwidth]{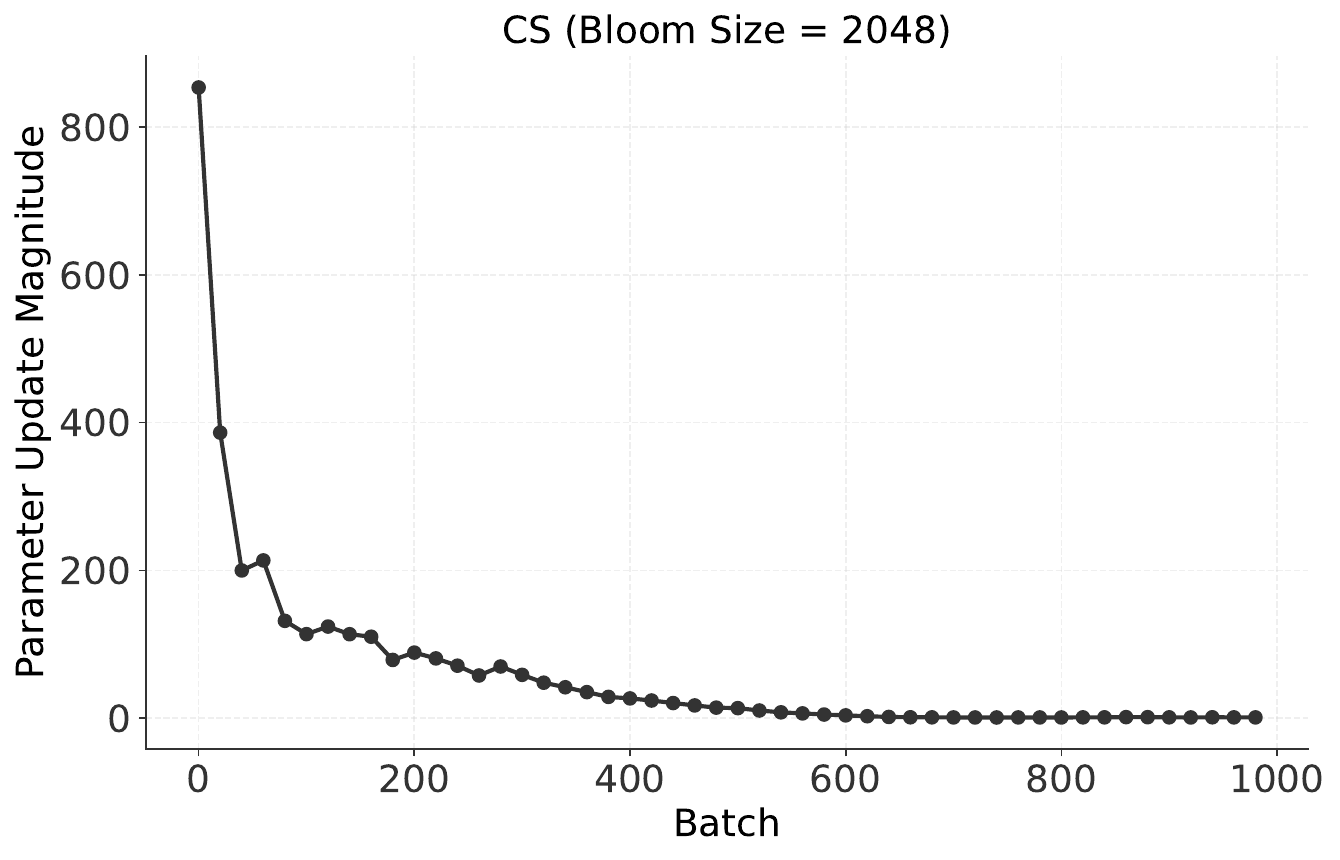}
    }
    \vspace{-0.3cm}
    % 第二行的两张图片
    \subfloat{
        \includegraphics[width=0.48\textwidth]{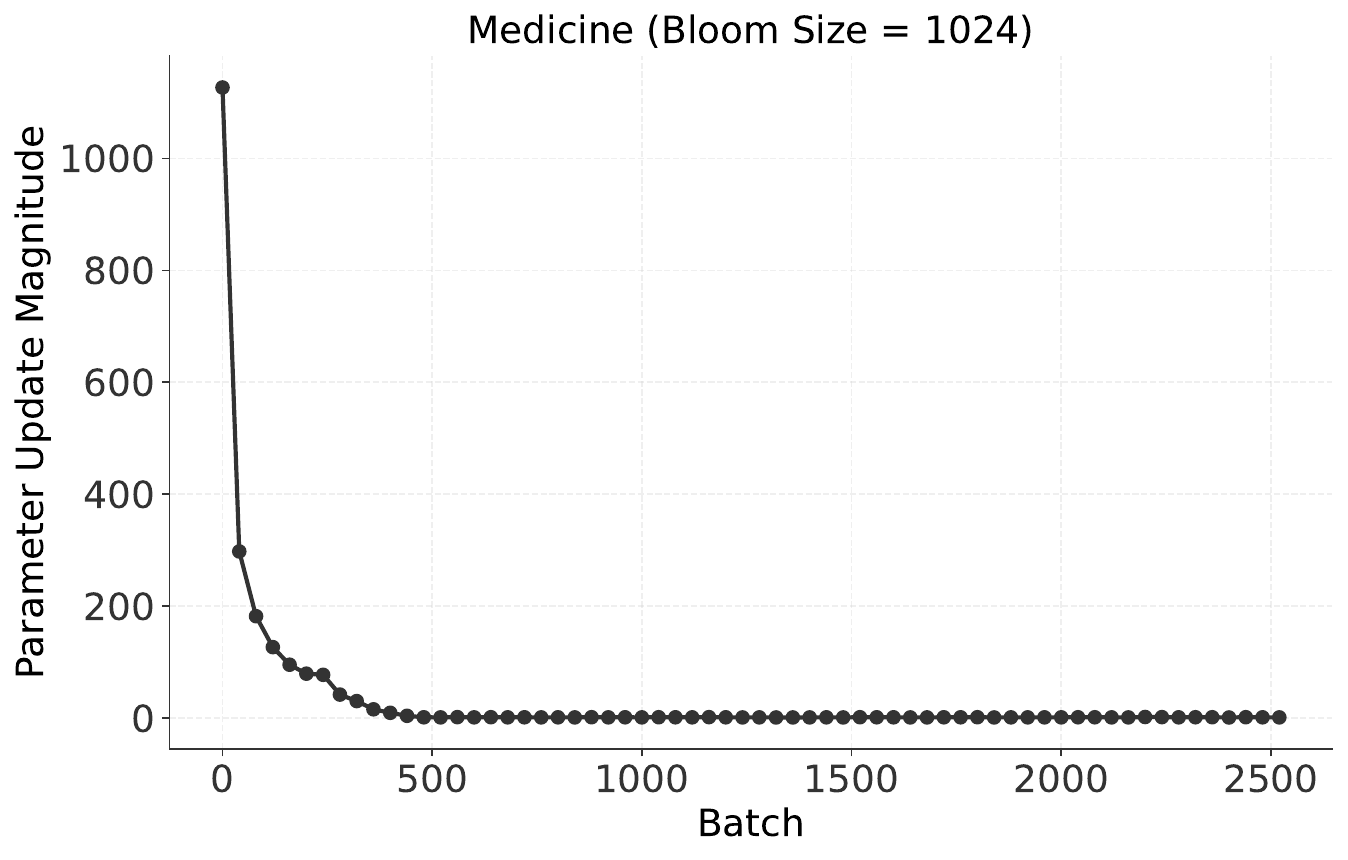}
    }
    \subfloat{
        \includegraphics[width=0.48\textwidth]{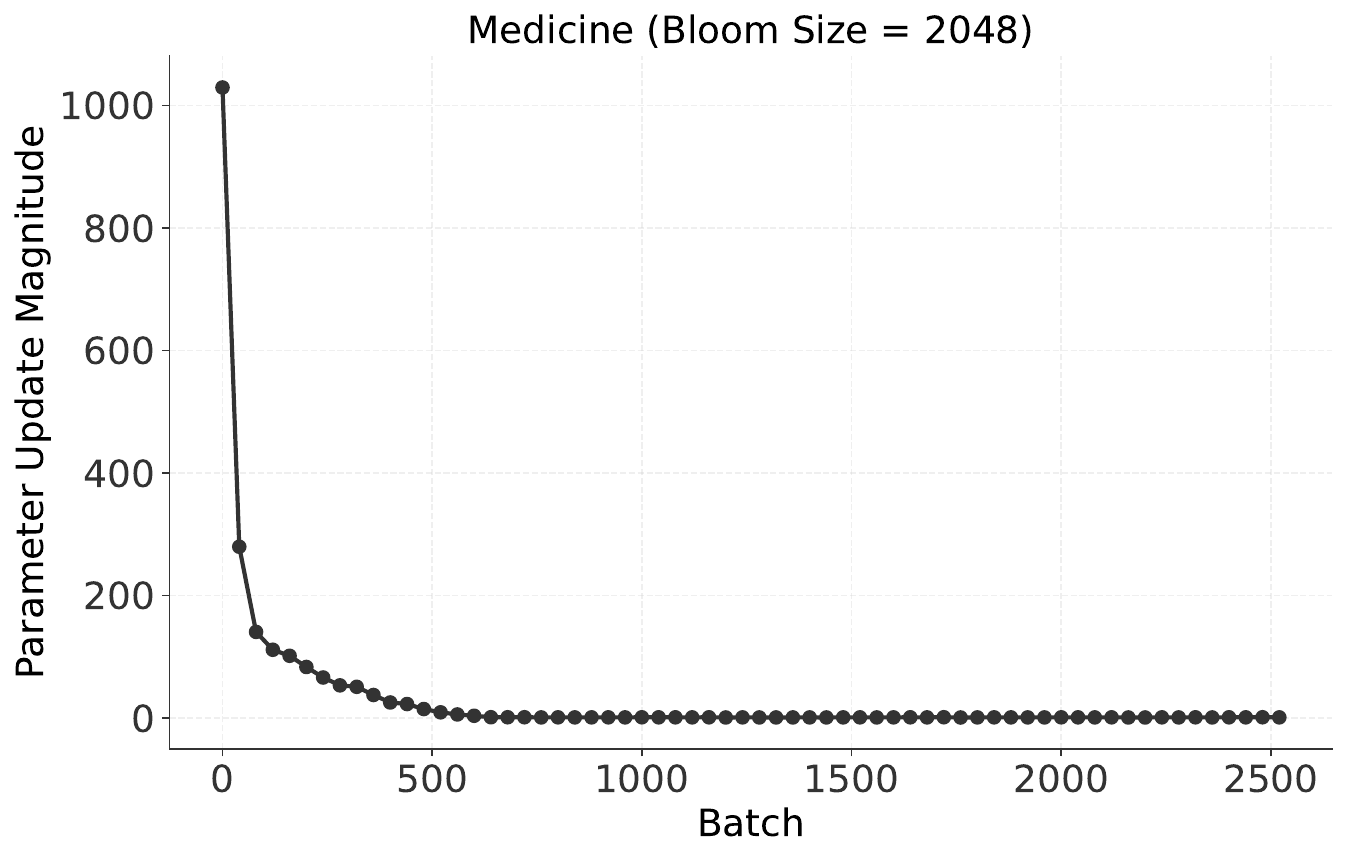}
    }
    \vspace{-0.3cm}
    % 第三行的两张图片
    \subfloat{
        \includegraphics[width=0.48\textwidth]{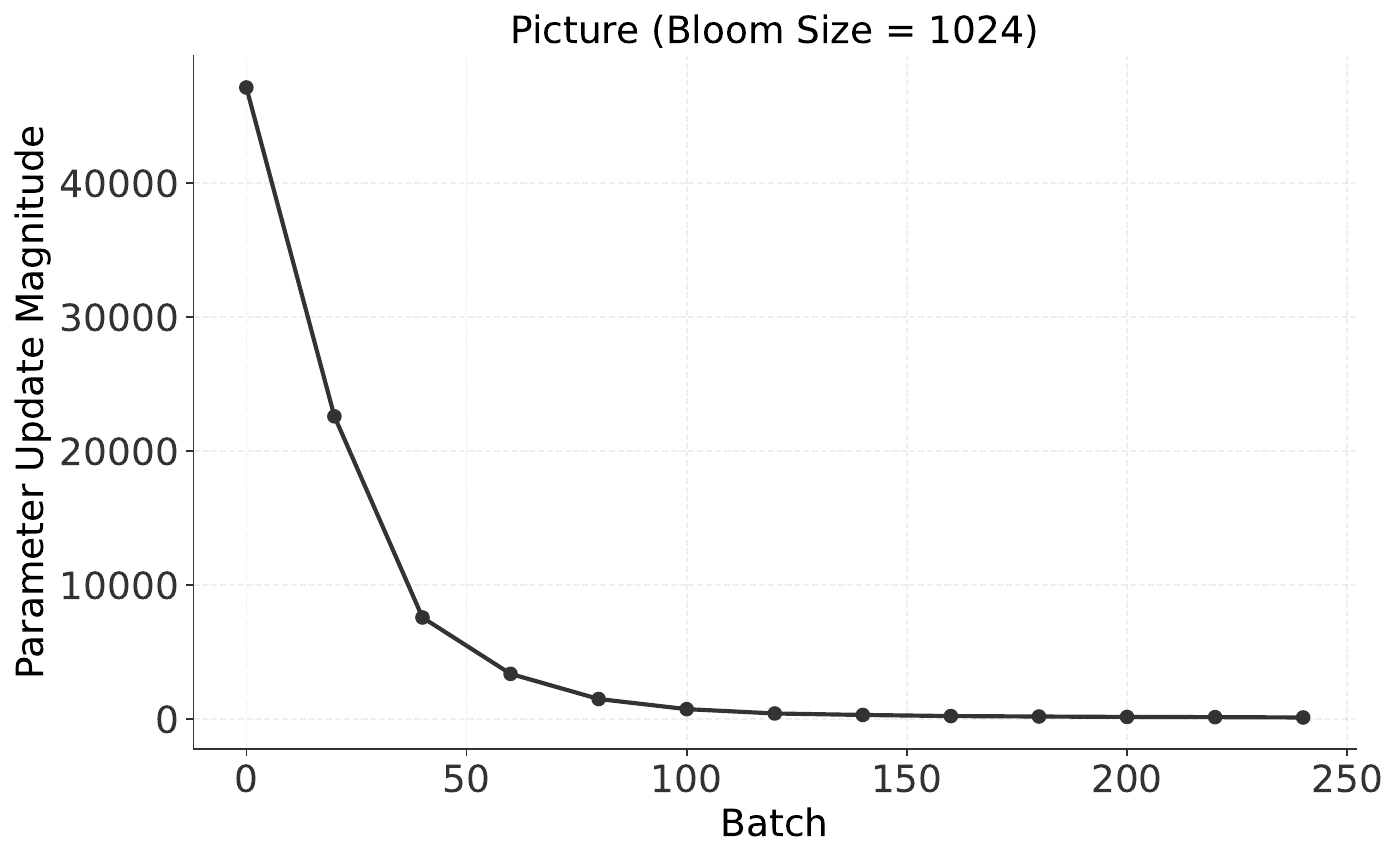}
    }
    \subfloat{
        \includegraphics[width=0.48\textwidth]{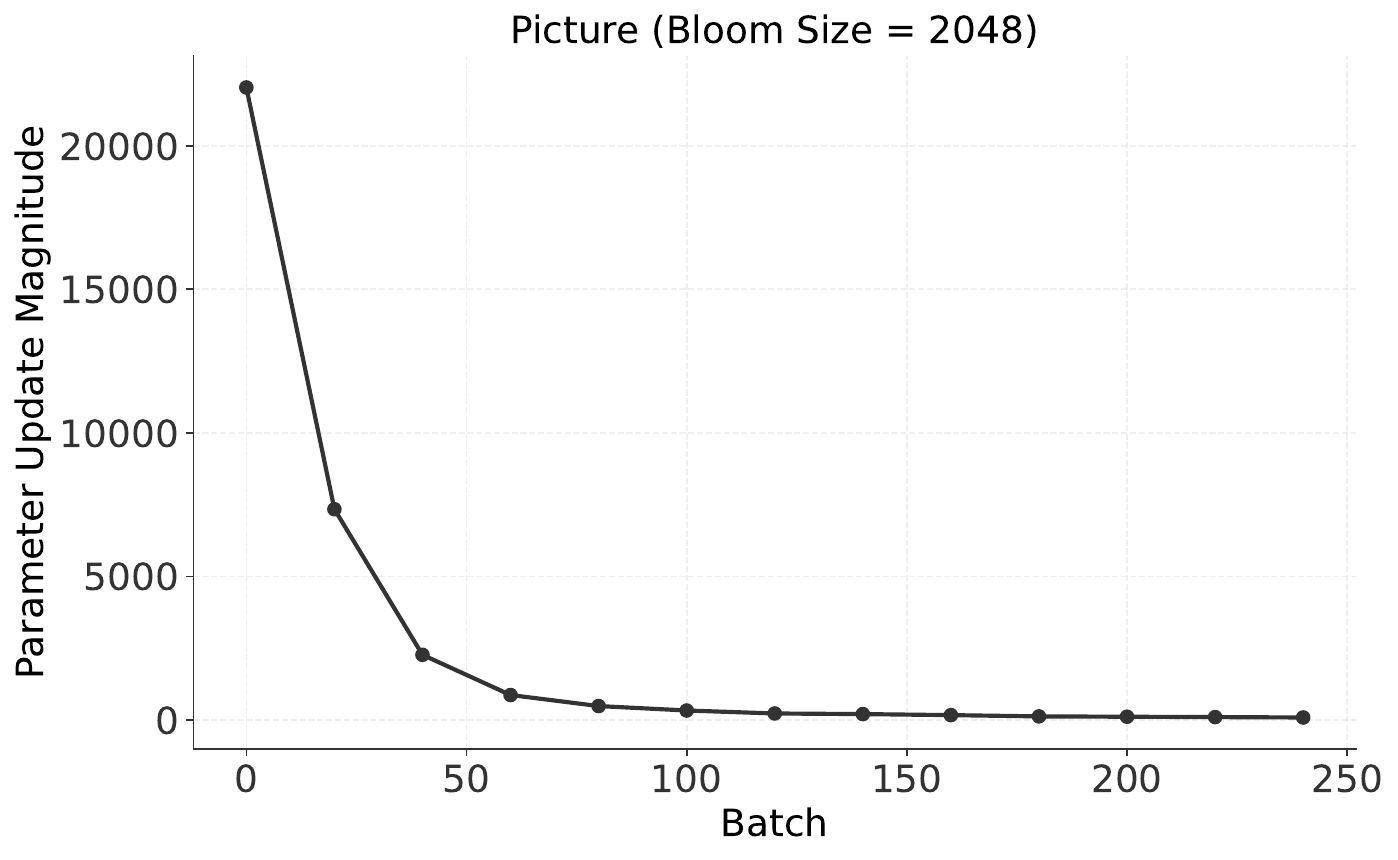}
    }
    
    \caption{Parameter Update Magnitude Across Batches in BioHash}
    \label{fig:biohash_iter}
\end{figure}

Figure \ref{fig:biohash_iter} illustrates the parameter update magnitude dynamics across three distinct datasets (\CS, \Medicine, and \Picture) under varying Bloom filter sizes (1024 and 2048 bits). The experimental results reveal several significant patterns:
\begin{itemize}[leftmargin=*, topsep=0pt]
    \item \textbf{Convergence Behavior:} Across all configurations, the parameter update magnitude exhibits a consistent decay pattern. The parameter updates show high magnitudes during initial batches, followed by a rapid decrease and eventual stabilization. This pattern indicates \BioHash algorithm's robust convergence properties regardless of the domain-specific data characteristics.

    \item \textbf{Filter Size Consistency:} The comparison between 1024 and 2048 Bloom filter configurations demonstrates remarkable consistency in convergence patterns. This observation suggests that \BioHash maintains stable performance characteristics independent of filter size, validating the robustness of our parameter update mechanism across different capacity settings.

    \item \textbf{Cross-Domain Consistency:} The similar convergence patterns observed across \CS, \Medicine, and \Picture datasets validate the algorithm's domain-agnostic nature. Despite the inherent differences in data distributions, \BioHash iteration mechanism maintains consistent performance characteristics.
\end{itemize}

Through extensive empirical analysis across diverse application scenarios, we validated BioHash's effectiveness as a core component of our system.   Our experiments demonstrated consistent and reliable convergence behavior under various environmental configurations.   The results confirmed that BioHash achieves stable parameter updates well before completing the full training process.

\section{Conclusions}
This work addresses the relatively underexplored challenge of vector set search. We propose \NaiveBioVSS and \BioVSS to enable efficient vector set search based on the Hausdorff distance. Theoretical analysis demonstrates the correctness of our algorithms. The dual-layer filtering mechanism efficiently eliminates irrelevant candidates, thereby minimizing computational costs. Experimental results demonstrate that our approach attains more than 50× acceleration over conventional linear scan techniques on million-scale datasets, while maintaining recall rates up to 98.9\%. Future research directions include extending the framework to accommodate additional distance metrics

\section*{Acknowledgement}
This work was supported by the National Key R\&D Program of China (2023YFB4503600), the National Natural Science Foundation of China (62202338, 62372337), and the Key R\&D Program of Hubei Province (2023BAB081)

\bibliographystyle{elsarticle-num} 

\bibliography{references}

\end{document}